\documentclass[10pt,twocolumn]{IEEEtran}

\usepackage[utf8]{inputenc}
\usepackage{amssymb,amsmath,amsfonts,amsthm,bm}
\usepackage{mathrsfs}
\usepackage{graphicx}
\usepackage{cite}
\usepackage{color}
\usepackage{lipsum}
\usepackage[bookmarks]{hyperref}
\usepackage{booktabs}
\usepackage[mathscr]{euscript}
\usepackage{soul}
\usepackage[scr]{rsfso}
\usepackage{tablefootnote}
\def\BibTeX{{\rm B\kern-.05em{\sc i\kern-.025em b}\kern-.08em
    T\kern-.1667em\lower.7ex\hbox{E}\kern-.125emX}}
\makeatletter
  \def\hrulefill{\leavevmode\leaders\hrule height 1pt\hfill\kern\z@}
\makeatother
\usepackage{wrapfig}
\usepackage{listings}
\newcommand{\Laplace}{\mathscr{L}}

\definecolor{gris245}{RGB}{245,245,245}
\definecolor{olive}{RGB}{50,140,50}
\definecolor{brun}{RGB}{175,100,80}
\usepackage{fancybox}

\hypersetup{
  unicode=false,
  colorlinks=true,       % false: boxed links; true: colored links
  linkcolor=black,       % color of internal links blue
  citecolor=black,       % color of links to bibliography blue
  filecolor=black,       % color of file links magenta
  urlcolor=black,
  bookmarksdepth=4
}

\newtheorem*{theorem}{Theorem}

\newtheorem*{corollary}{Corollary}

\begin{document}

\title{On the Exact Sum PDF and CDF of 
$\alpha$-$\mu$ Variates 
}

\author{Fernando Darío Almeida García, \textit{Member, IEEE},\\ Francisco Raimundo Albuquerque Parente, \textit{Graduate Student Member, IEEE}, \\ Michel Daoud Yacoub, \textit{Member, IEEE}, and José Cândido Silveira Santos Filho, \textit{Member, IEEE}

\thanks{
\hl{A Mathematica (Wolfram) implementation for computing the probability density function (PDF) and cumulative distribution function (CDF) of the sum of independent }\mbox{\hl{$\alpha$-$\mu$}}\hl{ random variables is available at the GitHub repository: }\href{https://github.com/Fernando-AlmeidaGit/On-the-Exact-Sum-PDF-and-CDF-of-alpha-mu-Variates.git}{\hl{https://github.com/Fernando-AlmeidaGit/On-the-Exact-Sum-PDF-and-CDF-of-alpha-mu-Variates.git.}}\hl{ The implementation is valid for $\alpha, \mu, \hat{r} > 0$ and $L \in \mathbb{N}$.}
The work of Fernando Darío Almeida García was supported by the São Paulo Research Foundation \mbox{(FAPESP)} under Grant \mbox{2021/03923-9}. The work of Francisco Raimundo Albuquerque Parente was supported by the São Paulo Research Foundation \mbox{(FAPESP)} under Grant \mbox{2018/25009-4.}

The authors are with the Wireless Technology Laboratory (WissTek), Department of Communications, School of Electrical and Computer Engineering, State University of Campinas \mbox{(UNICAMP)}, Campinas, SP 13083-852, Brazil (e-mail: \mbox{ferchoalmeida1@gmail.com}; \mbox{franciscoraparente@gmail.com}; \mbox{parente@ieee.org}; \mbox{mdyacoub@unicamp.br}; \mbox{jcssf@unicamp.br}).
}}

% \markboth{IEEE Transactions on Wireless Communications}%
% {Almeida \MakeLowercase{\textit{et al.}}:}

\maketitle

\begin{abstract}
The sum of random variables (RVs) appears extensively in wireless communications, at large, both conventional and advanced, and has been subject of longstanding research.
The statistical characterization of the referred sum is crucial to determine the performance of such communications systems.
Although efforts have been undertaken to unveil these sum statistics, e.g., probability density function (PDF) and cumulative distribution function (CDF), no general efficient nor manageable solutions capable of evaluating the exact sum PDF and CDF are available to date. The only formulations are given in terms of either the multi-fold Brennan's integral or the multivariate Fox $H$-function.  
Unfortunately, these methods are only feasible up to a certain number of RVs, meaning that  when the number of RVs in the sum increases, the computation of the sum PDF and CDF is subject to stability problems, convergence issues, or inaccurate results.
In this paper, we derive new, simple, exact formulations for the PDF and CDF of the sum of $L$ independent and identically distributed $\alpha$-$\mu$ RVs.
Unlike the available solutions, the computational complexity of our analytical expressions is independent of the number of summands.
Capitalizing on our unprecedented findings, we analyze, in exact and asymptotic manners, the performance of $L$-branch pre-detection equal-gain combining and maximal-ratio combining receivers over \mbox{$\alpha$-$\mu$} fading environments. 
The coding and diversity gains of the system for both receivers are analyzed and quantified. Moreover, numerical simulations show that the computation time reduces drastically when using our expressions, which are arguably the most efficient and manageable formulations derived so far.
\end{abstract}

\begin{IEEEkeywords}
$\alpha$-$\mu$ distribution, cumulative distribution function,
fading channels, probability density function, sums of random variables.
\end{IEEEkeywords}

\section{Introduction} 
\label{sec:intro}

\IEEEPARstart{S}{ums} of random variables (RVs) arise in different areas of knowledge, communications being one of them. In wireless and radar systems, for instance, these sums are found in applications involving performance analysis in signal detection, linear equalization, intersymbol interference, phase jitter, diversity-combining receivers, false-alarm rate detectors, secrecy capacity analysis~\cite{Halpern77,Beaulieu91,Garcia19,Annamalai00,Zhang97,Alouini01,Lei17}, etc. In these cases, sum statistics such as the probability density function (PDF), and the cumulative distribution function (CDF), among others, may be required.
In wireless communications, a number of general stochastic processes have proven useful to model the fading channel. In particular, the $\alpha$-$\mu$ process has gained considerable attention. In addition to being mathematically tractable, it adapts quite well to the statistical variations of the propagated signal and, in many instances, it yields much better results than the conventional models \cite{Leonardo15,Wu10}.
Many field measurements reported in the literature confirm the broad applicability of the $\alpha$-$\mu$ distribution in practical scenarios \cite{Dias09,Wu10,Boulogeorgos19,Karadimas10,Chong11,Michalopoulou12,Michalopoulou13}.
In \cite{Reig13}, the $\alpha$-$\mu$ distribution was placed within the context of composite distributions. Specifically, several composite distributions were employed to fit experimental data. For the analysis, both shadowing and multipath effects were considered. More interestingly, the authors showed that a single $\alpha$-$\mu$ distribution provides an excellent statistical fitting to the experimental data with the notorious advantage of its analytical simplicity. 
In \cite{Reig17, Marins19,Anjos19,Marins21}, it was shown that the $\alpha$-$\mu$ distribution can be used to model the fading channel in millimeter-wave communications. 
Furthermore, several works have explored the $\alpha$-$\mu$ model in emerging terahertz (THz) communications~\cite{Papasotiriou21,thz_dualhop,THz_pe,THz_relay}. In particular, the authors in \cite{Papasotiriou21} demonstrated that, among the unimodal distributions investigated therein, the $\alpha$-$\mu$ distribution proved best suited to model the fading signals in THz transmissions. In some circumstances, the statistics of the fading signal in a THz scenario appears in a multimodal manner. In such cases, multimodal distributions will certainly yield a better fit (e.g., the bimodal fluctuating two-ray model~\cite{thz_multimodal,ftr}).

The theory involving the statistics of the sum of fading variates has become a subject of interest since the dawn of wireless communications. On the other hand, it is widely known that key statistics of the exact sum of fading envelopes are either quite involving or unavailable. Thus, resorting to approximate, estimate, and exact methods has become a common practice. Next, we revisit some results obtained with such approaches.

\subsection{Motivation and Related works}

Notable works concerning the exact sum statistics for fading channels include \cite{Yilmaz09conf} for Weibull, \cite{Dharmawansa07,Rahman11} for Nakagami-$m$, and \cite{Rahama18,Kong2021} for $\alpha$-$\mu$. 
Although praiseworthy, these works make use of nested infinite sums, are expressed in terms of multi-fold integrals, or are given in terms of special functions that, unfortunately, are not yet available in any mathematical software.
For instance, recently, in \cite{Kong2021}, the authors employed the multivariate Fox $H$-function to obtain the sum PDF of independent, non-identically distributed (i.n.i.d.) $\alpha$-$\mu$ RVs. Regrettably, the multivariate Fox $H$-function is calculated through a multi-fold complex integration. This means that if one needs to compute the sum of $L$ $\alpha$-$\mu$ RVs, then the numerical computation of $L$ complex infinite-range integrals will be required, even for the independent and identically distributed (i.i.d.) case.
Accordingly, the computations in \cite{Yilmaz09conf,Dharmawansa07,Rahman11,Rahama18,Kong2021} become unfeasible as the number of summands increases.
A more general exact framework, applicable to any sum of non-negative variates, involves the so-called Brennan's  integral~\cite{Brennan59}. Here again, as in the previous cases, the computations are time-consuming, extremely costly, subject to convergence problems and to stability issues, and becomes unfeasible when the number of summands increases (say, above five).

Due to the cumbersome math behind the exact analysis, approximations are often used. (The readers are referred to \cite{Nakagami60,Karagiannidis05,Hu05,Filho06,Reig00,Costa08} for detailed discussions on the approximate sums of Rayleigh, Weibull, and Nakagami-$m$ RVs.)
In \cite{Costa08aplha}, the authors employed a moment-matching approach to approximate the sum PDF of i.i.d. $\alpha$-$\mu$ RVs by another $\alpha$-$\mu$ PDF.
In \cite{Perim20}, an optimal asymptotic-matching technique was used to approximate the sum PDF of positive i.n.i.d. RVs. 
In \cite{Parente19}, the results of \cite{Perim20} were further extended to address the approximate sum PDF of positive correlated RVs.
More recently, in~\cite{Payami21}, using the moment-matching method, the authors proposed two variable-order approximations to the sum of i.i.d. and i.n.i.d. $\alpha$-$\mu$ variates. 
To do so, the authors approximated the moments of the exact sum PDF to the moments of a truncated Puiseux series composed of a few identical $\alpha$-$\mu$ distributions. An alternative approach was proposed in \cite{ref_rev2_1,ref_rev2_2}, where the tails of the sum CDF of independent RVs were estimated using the concept of importance~sampling.

Despite the great effort to derive the $\alpha$-$\mu$ sum statistics, the search for efficient, tractable, and fast solutions is still an open problem in the communications community, a task even more challenging when considering the exact~approach.

\subsection{Contributions}

In this work, new, handy, exact formulas for the sum PDF and the sum CDF of i.i.d. $\alpha$-$\mu$ RVs is proposed. These are arguably the most efficient solutions reported in the open literature so far.

Next, we outline the main contributions of this paper.
\begin{itemize}
    \item Derivation of exact novel formulations for the PDF and the CDF of the sum of $L$ i.i.d. $\alpha$-$\mu$ variates.  
    As shall be seen, the proposed expressions are fast, compact, and, more importantly, their mathematical complexity remains the same independently of the number of summands.
    \item  Derivation of the exact formulations for the performance analysis of $L$-branch pre-detection equal-gain combining (EGC) and maximal-ratio combining (MRC) receivers operating over $\alpha$-$\mu$ fading environments. More precisely, exact solutions for the average symbol error rate (ASER) and outage probability (OP) in EGC and MRC are obtained.
    Also, the coding and diversity gains of the system for both receivers are analyzed and quantified.
    \item Asymptotic analysis so as to provide an intuitive understanding on how the distribution parameters impact the system performance. 
    It is worth mentioning that the exact and asymptotic solutions presented herein for the ASER and OP are also original contributions, being simpler and less costly than those found elsewhere (e.g., in \cite{Costa08aplha,Payami21,ElAyadi14}, and \cite{BenIssaid18}).
\end{itemize}
%=============================================

\subsection{Paper Structure and Notation}

The rest of this paper is organized as follows. 
Section~\ref{sec: Problem Formulation} describes the problem statement and preliminaries. Section~\ref{sec: Proposed Solution} derives the sum  PDF and the sum CDF of i.i.d. $\alpha$-$\mu$ RVs. Section \ref{sec: Equal Gain Combibing diversity Receivers} analyzes the performance of EGC and MRC diversity receivers.
Section~\ref{sec: Numerical Results} discusses representative numerical results. Finally, Section~\ref{sec: Conclusions} concludes the~paper.

In what follows, $\text{Pr}[\cdot]$ denotes probability; $f_{(\cdot)}(\cdot)$, PDF; $F_{(\cdot)}(\cdot)$, CDF; $\mathbb{E} \left[ \cdot \right]$, expectation; $\mathbb{V} \left[ \cdot  \right]$, variance; $\Laplace \left\{ \cdot \right\}$, the Laplace transform; $\Laplace^{-1}\left\{ \cdot \right\}$, the inverse Laplace transform; $\left| \cdot \right|$, absolute value; $*$, convolution; $j=\sqrt{-1}$, the imaginary unit; $\mathbb{N}_0$, the set of natural numbers including zero; $\mathbb{Z}^{+}$, the set of positive integer numbers; $\mathbb{C}$, the set of complex~numbers; and $\simeq$, ``asymptotically equal to around zero,'' i.e., ${h(x)\simeq g(x) \iff \lim\limits_{x\to 0} \, \frac{h (x)}{g (x)}=1}$;

\section{Problem Statement and Preliminaries}
\label{sec: Problem Formulation}
The problem tackled here is an age-old one, a challenge that has been the subject of a number of works by several researchers. More particularly, the question addressed is to find exact statistics, namely PDF, CDF, and applications thereof, for the sum of fading envelopes. We solve this problem for a rather general distribution: the $\alpha$-$\mu$ model.

Let $R$ be the sum of $L$ i.i.d. $\alpha$-$\mu$ RVs $R_n$, i.e., 
% \small
\label{sec:problem}
\begin{equation}\label{eq: sum}
    R= \sum _{n=1}^{L} R_n.
\end{equation}
\normalsize

A signal with envelope  $\{R_n\}_{n=1}^L$ modeled by the $\alpha$-$\mu$ distribution has a PDF given as~\cite{Yacoub07}
\begin{align}
    \label{eq: PDF Xn}
    f_{R_n}(r_n)=& \frac{\alpha  \mu ^{\mu } r_n^{\alpha  \mu -1}}{\Gamma (\mu ) \hat{r}^{\alpha  \mu }} \exp \left(-\mu  \left(\frac{r_n}{\hat{r}}\right)^{\alpha }\right), 
\end{align}
\normalsize
where $\alpha>0$ is a nonlinearity parameter, $\hat{r}\triangleq \sqrt[\alpha]{\mathbb{E} \left[ R_{n}^\alpha\right]}$ is the $\alpha$-root mean value, $\mu>0$ is the inverse of the normalized variance of $R_{n}^\alpha$, i.e.,
\begin{align}
    \label{}
    \mu= \frac{\mathbb{E}^2 \left[ R_{n}^\alpha\right] }{\mathbb{V} \left[R_{n}^\alpha \right]},
\end{align}
and  $\Gamma(\cdot)$ is the gamma function~\cite[eq. 6.1.1]{abramowitz72}.

The CDF of $R_n$ is given by \cite{Yacoub07}
\begin{align}
    \label{}
    F_{R_n}(r_n) = \frac{\gamma \left( \mu, \mu r_{ n}^{\alpha} / \hat{r}^{\alpha} \right)}{\Gamma(\mu) },
\end{align}
where $\gamma (\cdot,\cdot)$ is the lower incomplete gamma function~\cite[eq. (8.2.1)]{Olver10}.

The $\alpha$-$\mu$ distribution is a general and versatile fading model
that includes several other important small-scale fading models as special cases, such as one-sided
Gaussian, negative exponential, Rayleigh, Gamma, Weibull, and Nakagami-$m$ distributions. 
The solutions proposed next for the sum PDF and CDF are exact, simple, easy to compute, and unprecedented in the literature.

\section{Exact Sum Statistics}
\label{sec: Proposed Solution}

In this section, we derive exact, fast, and manageable solutions for the PDF and the CDF of the sum of i.i.d. $\alpha$-$\mu$ RVs. These solutions are presented in the following theorem and corollary.

\begin{theorem}
The PDF for the sum in \eqref{eq: sum} is given by
\begin{align}    \label{eq: Final PDF}
    f_{R}(r) = \left(\frac{\alpha  \mu ^{\mu }}{\Gamma (\mu ) \hat{r}^{\alpha  \mu }}\right)^L  \sum _{i=0}^{\infty } \frac{\delta_i r^{\alpha  i+\alpha  \mu  L-1}}{\Gamma (i \alpha +L \mu  \alpha )},
\end{align}
where the coefficients $\delta_i$ can be obtained by the following recursive formulas: 
\begin{subequations}
\label{eq: Coefficients}
\begin{align}
    \delta_0=& \Gamma (\alpha  \mu )^L \\
    \delta_i=& \frac{1 }{i \Gamma (\alpha  \mu )}\sum _{l=1}^i \frac{\delta_{i-l}(l L+l-i)  \Gamma (\alpha  (l+\mu )) \left(-\mu  \left(\frac{1}{\hat{r}}\right)^{\alpha }\right)^l}{l!}.
\end{align}
\end{subequations}
\end{theorem}

\begin{proof}
The corresponding proof can be found in Appendix \ref{app: PDF R} and relies on two fundamental principles of complex analysis and calculus, namely, 
complex integration and differential equations.

\end{proof}

The absolute convergence of \eqref{eq: Final PDF} is confirmed in Appendix \ref{app: Absolute Convergence}.
As a simple check, for $L=1$, \eqref{eq: Final PDF} reduces to the $\alpha$-$\mu$ distribution, as shown in Appendix \ref{app: singel alpha-mu}.
Furthermore, if the series in \eqref{eq: Final PDF} is truncated up to the $\mathcal{N}_{\text{T}}$ terms, then its associated truncation error can be obtained as
\begin{align}
    \label{eq: truncation def PDF}
    \epsilon_{f} =\left(\frac{\alpha  \mu ^{\mu }}{\Gamma (\mu ) \hat{r}^{\alpha  \mu }}\right)^L  \sum _{i=\mathcal{N}_{\text{T}}}^{\infty } \frac{\delta_i r^{\alpha  i+\alpha  \mu  L-1}}{\Gamma (i \alpha +L \mu  \alpha )}.
\end{align}
In Table \ref{tab: PDF Accuracy}, we show that few terms (between 12 and 15) were required to ensure an outstanding PDF accuracy of less than $10^{-13}$.

In practice, it is most useful to provide an upper bound for the truncation error in \eqref{eq: truncation def PDF}. 
However, because of the increasing or decreasing behavior of the coefficients $\delta_i$ (see Appendix \ref{app: Absolute Convergence}), it is more convenient to obtain two different upper bounds depending on the values of $\alpha$, i.e., one bound for $0<\alpha<1$, denoted by $\mathcal{B}_{f}^\dagger$, and another for $\alpha \geq1$, denoted by $\mathcal{B}_{f}^*$.

$\mathcal{B}_{f}^\dagger$ can be found in closed form as (see Appendix \ref{app: Truncation Bounds}) 
\begin{align}
    \label{eq: bound k<1}
    \nonumber \mathcal{B}_{f}^\dagger = & \, r^{-1} \left(\frac{\alpha  \mu ^{\mu } \Gamma (\alpha  \mu ) \left(\frac{r}{\hat{r}}\right)^{\alpha  \mu }}{\Gamma (\mu )}\right)^L \\
    \nonumber & \times \Bigg[ E_{\alpha ,\alpha  \mu  L}\left(2 \mu  L \left(\frac{r}{\hat{r}}\right)^{\alpha } \Gamma (\mu  \alpha +\alpha )\right)  \\
    & \ \ \ \ -\sum _{i=0}^{\mathcal{N}_{\text{T}}-1} \frac{\left(2 \mu  L \left(\frac{r}{\hat{r}}\right)^{\alpha } \Gamma (\mu  \alpha +\alpha )\right)^i}{\Gamma (i \alpha +L \mu  \alpha )}  \Bigg].
\end{align}
Notice in \eqref{eq: bound k<1} that the sum within the brackets is a lower truncated version of the Mittag-Leffler function on the left.
Therefore, when $\mathcal{N}_{\text{T}} \to \infty$, then the Mittag-Leffler function and the sum will cancel each other out and, consequently, $\mathcal{B}_{f}^\dagger$ will go to zero.

$\mathcal{B}_{f}^*$ can also be obtained in closed form as (see Appendix~\ref{app: Truncation Bounds}) 
\begin{align}
    \label{eq: bound k>=1}
    \nonumber \mathcal{B}_{f}^* = & \frac{2 L}{r} \left(\frac{\alpha  \mu ^{\mu } \Gamma (\alpha  \mu ) \left(\frac{r}{\hat{r}}\right)^{\alpha  \mu }}{\Gamma (\mu )}\right)^L \\
    & \times \frac{\exp \left(\mu  \left(\frac{r}{\hat{r}}\right)^{\alpha }\right) \gamma \left(\mathcal{N}_{\text{T}},\mu  \left(\frac{r}{\hat{r}}\right)^{\alpha }\right)}{\Gamma \left(\mathcal{N}_{\text{T}}\right)}.
\end{align}
Notice in \eqref{eq: bound k>=1} that since $\underset{\mathcal{N}_{\text{T}} \to \infty }{\text{lim}}\frac{\gamma (\mathcal{N}_{\text{T}},a)}{\Gamma (\mathcal{N}_{\text{T}})}=0$, then $\mathcal{B}_{f}^*$ approaches zero as $\mathcal{N}_{\text{T}}$ goes to infinity.

It is important to highlight that \eqref{eq: bound k<1} and \eqref{eq: bound k>=1} can be used to establish a sufficient number of terms to guarantee a required PDF accuracy.
For example, in Table \ref{tab: PDF Accuracy}, we show that 15 terms are sufficient (i.e., they can be less) to ensure a PDF accuracy of $10^{-11}$.

%=================================================
\begin{table*}[t!]
\centering
\caption{PDF Accuracy Analysis}
\label{tab: PDF Accuracy}
\begin{tabular}{l c c c c}
\toprule
\multicolumn{1}{c}{PDF parameters} & $f_R(r)$ & $\mathcal{N}_{\text{T}}$ & Truncation error, $\epsilon_f$  & Truncation bounds, $\mathcal{B}_{f}^\dagger$ or $\mathcal{B}_{f}^*$  \\ \hline
$\alpha=0.8$, $\mu=0.2$, $r=2$, $\hat{r}=5$, $L=3$ & $0.06218$ & 13 & $1.94428 \times 10^{-26}$ & $5.44576 \times 10^{-11}$  \\ \hline
$\alpha=1.2$, $\mu=0.5$, $r=2$, $\hat{r}=1$, $L=3$ & $0.24861$ & 15 & $4.62796 \times 10^{-14}$ & $2.49407 \times 10^{-11}$  \\ \hline
$\alpha=0.7$, $\mu=0.1$, $r=2$, $\hat{r}=7$, $L=4$ & $0.04389$  & $13$ & $2.46361 \times 10^{-31}$ & $1.15486 \times 10^{-11}$ \\ \hline
$\alpha=1.5$, $\mu=0.7$, $r=2$, $\hat{r}=2$, $L=4$ & $0.02492$ & 12 & $3.14923 \times 10^{-16}$ & $7.20537 \times 10^{-11}$  \\ \hline
$\alpha=0.9$, $\mu=0.7$, $r=3$, $\hat{r}=10$, $L=5$ & $0.00067$ & 15 & $2.81411 \times 10^{-26}$ & $2.13662 \times 10^{-11}$  \\ \hline
$\alpha=1.7$, $\mu=1.0$, $r=3$, $\hat{r}=3$, $L=5$ & $0.00013$ & 15 & $1.27159 \times 10^{-22}$ & $2.39054 \times 10^{-11}$  \\ 
\bottomrule  
\end{tabular}
\end{table*}
%=================================================
%=================================================
\begin{table*}[t!]
\centering
\caption{CDF Accuracy Analysis}
\label{tab: CDF Accuracy}
\begin{tabular}{l c c c c}
\toprule
\multicolumn{1}{c}{CDF parameters} & $F_R(r)$ & $\mathcal{N}_{\text{T}}$ & Truncation error, $\epsilon_F$  & Truncation bounds, $\mathcal{B}_{F}^\dagger$ or $\mathcal{B}_{F}^*$  \\ \hline
$\alpha=0.8$, $\mu=0.2$, $r=2$, $\hat{r}=5$, $L=3$ & $0.27666$ & 13 & $3.33080 \times 10^{-27}$ & $9.94579 \times 10^{-11}$  \\ \hline
$\alpha=1.2$, $\mu=0.5$, $r=2$, $\hat{r}=1$, $L=3$ & $0.42717$ & 15 & $4.42165 \times 10^{-15}$ & $4.98814 \times 10^{-11}$  \\ \hline
$\alpha=0.7$, $\mu=0.1$, $r=2$, $\hat{r}=7$, $L=4$ & $0.32405$  & $11$ & $1.41096 \times 10^{-29}$ & $2.98271 \times 10^{-11}$ \\ \hline
$\alpha=1.5$, $\mu=0.7$, $r=2$, $\hat{r}=2$, $L=4$ & $0.01330$ & 13 & $9.61472 \times 10^{-19}$ & $7.72864 \times 10^{-11}$  \\ \hline
$\alpha=0.9$, $\mu=0.7$, $r=3$, $\hat{r}=10$, $L=5$ & $0.00068$ & 14 & $3.02419 \times 10^{-25}$ & $2.30528 \times 10^{-11}$  \\ \hline
$\alpha=1.7$, $\mu=1.0$, $r=3$, $\hat{r}=3$, $L=5$ & $0.00005$ & 15 & $1.07039 \times 10^{-23}$ & $7.17163 \times 10^{-11}$  \\  
\bottomrule  
\end{tabular}
\end{table*}
%=================================================

\begin{corollary}
The CDF for the sum in \eqref{eq: sum} is given by
\begin{align}
    \label{eq: Final CDF}
    F_{R}(r) = \left(\frac{\alpha  \mu ^{\mu }}{\Gamma (\mu ) \hat{r}^{\alpha  \mu }}\right)^L \sum _{i=0}^{\infty } \frac{\delta_i r^{\alpha  i+\alpha  \mu  L}}{\Gamma (i \alpha +L \mu  \alpha +1)},
\end{align}
where the coefficients $\delta_i$ are given by \eqref{eq: Coefficients}.
\end{corollary}

\begin{proof}
See Appendix \ref{app: CDF R}.
\end{proof}

To prove the absolute convergence of \eqref{eq: Final CDF}, one can use the same approach as the one shown in~Appendix~\ref{app: Absolute Convergence}. 
Moreover, if the series in \eqref{eq: Final CDF} is truncated up to the $\mathcal{N}_{\text{T}}$ terms, then its associated truncation error can be expressed as
\begin{align}
    \label{eq: truncation def CDF}
    \epsilon_{F} =\left(\frac{\alpha  \mu ^{\mu }}{\Gamma (\mu ) \hat{r}^{\alpha  \mu }}\right)^L  \sum _{i=\mathcal{N}_{\text{T}}}^{\infty } \frac{\delta_i r^{\alpha  i+\alpha  \mu  L}}{\Gamma (i \alpha +L \mu  \alpha +1)}.
\end{align}
In Table \ref{tab: CDF Accuracy}, we show that no more that 15 terms were needed to ensure a remarkably CDF accuracy of less than $10^{-14}$.

Similarly to \eqref{eq: bound k<1} and \eqref{eq: bound k>=1}, the truncation error in \eqref{eq: truncation def CDF} error can be bounded considering two different bounds, one bound for $0<\alpha<1$, denoted by $\mathcal{B}_{F}^\dagger$, and another for $\alpha \geq1$, denoted by $\mathcal{B}_{F}^*$.
Thus, utilizing the same mathematical procedure as in Appendix \ref{app: Truncation Bounds}, $\mathcal{B}_{F}^\dagger$ and $\mathcal{B}_{F}^*$ can be obtained respectively as
\begin{align}
    \label{eq: bound k<1 CDF}
    \nonumber \mathcal{B}_{F}^\dagger = & \,  \left(\frac{\alpha  \mu ^{\mu } \Gamma (\alpha  \mu ) \left(\frac{r}{\hat{r}}\right)^{\alpha  \mu }}{\Gamma (\mu )}\right)^L \\
    \nonumber & \times \Bigg[ E_{\alpha ,\alpha  \mu  L+1}\left(2 \mu  L \left(\frac{r}{\hat{r}}\right)^{\alpha } \Gamma (\mu  \alpha +\alpha )\right)  \\
    & \ \ \ \ -\sum _{i=0}^{\mathcal{N}_{\text{T}}-1} \frac{\left(2 \mu  L \left(\frac{r}{\hat{r}}\right)^{\alpha } \Gamma (\mu  \alpha +\alpha )\right)^i}{\Gamma (i \alpha +L \mu  \alpha +1)}  \Bigg] \\ \label{eq: bound k>=1 CDF}
    \nonumber \mathcal{B}_{F}^* = &  2 L \left(\frac{\alpha  \mu ^{\mu } \Gamma (\alpha  \mu ) \left(\frac{r}{\hat{r}}\right)^{\alpha  \mu }}{\Gamma (\mu )}\right)^L \\
    & \times \frac{\exp \left(\mu  \left(\frac{r}{\hat{r}}\right)^{\alpha }\right) \gamma \left(\mathcal{N}_{\text{T}},\mu  \left(\frac{r}{\hat{r}}\right)^{\alpha }\right)}{\Gamma \left(\mathcal{N}_{\text{T}}\right)}.
\end{align}
Applying the same reasoning as in \eqref{eq: bound k<1} and \eqref{eq: bound k>=1}, it can be shown that \eqref{eq: bound k<1 CDF} and \eqref{eq: bound k>=1 CDF} also go to zero as $\mathcal{N}_{\text{T}}$ goes to infinity.
Also, notice \eqref{eq: bound k<1 CDF} and \eqref{eq: bound k>=1 CDF} can be employed to establish a sufficient number of terms given a required CDF accuracy.
For instance, in Table \ref{tab: CDF Accuracy}, we show that 15 terms are sufficient (i.e., they can be less) to guarantee a CDF accuracy of $10^{-11}$.

\section{Applications  to EGC and MRC Receivers}
\label{sec: Equal Gain Combibing diversity Receivers}

As an application example, we now use our findings to investigate the performance of EGC and MRC diversity receivers operating over $\alpha$-$\mu$ fading channels.

The instantaneous signal-to-noise ratio (SNR) in an $L$-branch pre-detection EGC diversity system with equal noise levels is given~by
\begin{align}
    \label{eq: SNR def}
    \Psi_{\text{EGC}}=\frac{E_s}{L N_0} \left(\sum _{n=1}^{L} R_n \right)^2=\frac{E_s}{L N_0} R_{\text{EGC}}^2,
\end{align}
where $\{R_n\}_{n=1}^L$ is the set of i.i.d. $\alpha$-$\mu$ RVs, $E_s$ is the mean energy of the transmitted symbol, $N_0$ is the power spectral density of the noise, and $E_s/N_0$ is the mean SNR per symbol~\cite{bookAlouini}.

On the other hand, the instantaneous SNR in an $L$-branch pre-detection MRC diversity system with equal noise levels is defined as
\begin{align}
    \label{eq: SNR def MRC}
    \Psi_{\text{MRC}}=\frac{E_s}{L N_0} \left( \sum _{n=1}^L R_n^2 \right)= \frac{E_s}{L N_0} R_{\text{MRC}}.
\end{align}

To determine the statistics of $\Psi_{\text{EGC}}$ and $\Psi_{\text{MRC}}$, we first proceed to find the PDF of $R_{\text{EGC}}$ and $R_{\text{MRC}}$.

The PDF of $R_{\text{EGC}}$ and $R_{\text{MRC}}$ can be easily obtained from \eqref{eq: Final PDF} after a conventional transformation of variables as
\begin{align}    \label{eq: PDF EGC MRC}
    f_{R_\nu}(r) = \left(\frac{\alpha  \mu ^{\mu }}{\vartheta_\nu \Gamma (\mu ) \hat{r}^{\alpha  \mu }}\right)^L  \sum _{i=0}^{\infty } \frac{\delta_i r^{\frac{\alpha  i+\alpha  \mu  L}{\vartheta_\nu}  -1}}{\Gamma ( \frac{i \alpha +L \mu  \alpha}{\vartheta_\nu})},
\end{align}
in which $\nu \in \left\{ \text{EGC}, \text{MRC}\right\}$, 
$\vartheta_{{\text{EGC}}} =1$, and $\vartheta_{{\text{MRC}}} =2$.

From \eqref{eq: PDF EGC MRC} and after another straightforward transformation of variables, the PDF of $\Psi_{\nu}$ can obtained as
\begin{align}
    \label{eq: PDF psi nu}
    \nonumber f_{\Psi_{\nu}}(\psi)= & \frac{\vartheta_\nu L N_0}{2 E_s} \left(\frac{\alpha  \mu ^{\mu }}{\vartheta_\nu \Gamma (\mu ) \hat{r}^{\alpha  \mu }}\right)^L \\
    & \times \sum _{i=0}^{\infty } \frac{\delta_i \left(\frac{L N_0 \psi }{E_s}\right)^{\frac{1}{2} (\alpha  i+\alpha  \mu  L)-1}}{\Gamma ( \frac{i \alpha +L \mu  \alpha }{\vartheta_\nu })}.
\end{align}
 
Now, from \eqref{eq: PDF psi nu}, the CDF of $\Psi_{\nu}$ can be calculated as
\begin{align}
    \label{eq: CDF Phi}
    \nonumber F_{\Psi_\nu}(\psi) \triangleq & \int_{0}^{\psi}  f_{\Psi_\nu}(u) \text{d}u \\
    =&  \left(\frac{\alpha  \mu ^{\mu }}{\vartheta_\nu \Gamma (\mu ) \hat{r}^{\alpha  \mu }}\right)^L \sum _{i=0}^{\infty } \frac{\delta_i \left(\frac{N_0 L \psi }{E_s}\right)^{\frac{1}{2} (\alpha  i+\alpha  \mu  L)}}{ \Gamma (\frac{i \alpha +L \mu  \alpha }{\vartheta_\nu}+1  )}.
\end{align}

\subsubsection{Average Symbol Error Rate}
\label{sec: Average SER}
The ASER for a pre-detection EGC/MRC receiver is expressed as~\cite[eq. (9.61)]{bookAlouini}
\begin{align}
    \label{eq: Pe def}
    P_{\textmd{e}_\nu}=\int^{\infty}_{0} Q\left( \sqrt{2 \mathcal{G} \psi}\right)  f_{\Psi_\nu}(\psi) \text{d} \psi,
\end{align}
where $Q\left(\cdot \right)$ is
the Gaussian $Q$-function~\cite[eq. (4.1)]{bookAlouini}, and $\mathcal{G}$ is a modulation-dependent parameter such that $\mathcal{G}=1$ for binary phase-shift keying (BPSK), $\mathcal{G}=1/2$ for orthogonal BPSK, and $\mathcal{G}=0.715$ for BPSK with minimum correlation~\cite{bookProakis95}.
Other values of $\mathcal{G}$ can also be used for other modulation schemes~\cite{bookAlouini}.

Using the alternative representation of the Gaussian complementary cumulative function~\cite[eq. (4.2)]{bookAlouini}, we can rewrite \eqref{eq: Pe def} as follows:
\begin{align}
    \label{eq: Pe def 2}
    P_{\textmd{e}_\nu}=\frac{1}{2}\int^{\infty}_{0} \text{erfc}\left( r \sqrt{\frac{\mathcal{G} E_s}{L N_0}}\right)  f_{R_\nu}(r) \text{d} r,
\end{align}
where $\text{erfc} \left( \cdot \right)$ is the complementary error function~\cite[eq. (06.27.07.0001.01)]{Mathematica}. Substituting \eqref{eq: Final PDF} into \eqref{eq: Pe def 2}, and changing the order of integration, we obtain
\begin{align}
    \label{}
    \nonumber P_{\textmd{e}_\nu}= &\frac{1}{2} \left(\frac{\alpha  \mu ^{\mu }}{\vartheta_\nu \Gamma (\mu ) \hat{r}^{\alpha  \mu }}\right)^L \sum _{i=0}^{\infty } \frac{\delta_i}{\Gamma ( \frac{i \alpha +L \mu  \alpha }{\vartheta_\nu} )} \\
    & \times \int_0^{\infty } \text{erfc}\left(r \sqrt{\frac{E_s \mathcal{G}}{L N_0}}\right) r^{ \frac{\alpha  i+\alpha  \mu  L}{\vartheta_\nu} -1} \, \text{d}r.
\end{align}  
The previous integral can be solved with the aid of~\cite[eq. (1.5.2.1)]{prudnikov98Vol1}. Then, after simplifications, we finally attain
\begin{align}
    \label{eq: Pe final}
    P_{\textmd{e}_\nu}= &\frac{1}{2 } \left(\frac{\alpha  \mu ^{\mu }}{\vartheta_\nu \Gamma (\mu ) \hat{r}^{\alpha  \mu }}\right)^L \sum _{i=0}^{\infty } \frac{\delta_i \left(\frac{4 E_s \mathcal{G}}{L N_0 }\right)^{-\frac{\alpha  i+\alpha  \mu  L}{2 \vartheta_\nu} }}{ \Gamma ( \frac{i \alpha +L \mu  \alpha}{2 \vartheta_\nu}+1  )}.
\end{align}

An important region for the performance analysis of communications systems is the high-SNR regime. To address this, we take the most significant term in the summation of \eqref{eq: Pe final}, which is obtained when $i=0$. Thus, considering only the first term, the asymptotic ASER can be obtained~as
\begin{align}
    \label{eq: Pe Asymptotic}
    P_{\textmd{e}_\nu}  \simeq & \left( \mathcal{C}_{\textmd{e}_\nu} \frac{E_s}{N_0}  \right)^{-\mathcal{D}_{\textmd{e}_\nu}},
\end{align}
where $\mathcal{D}_{\textmd{e}_\nu} =(\alpha  \mu  L)/2 \vartheta_\nu$ is the diversity gain and
\begin{align}
    \label{eq: cod gain ABER} \mathcal{C}_{\textmd{e}_\nu}  = \frac{4 \mathcal{G}}{L} \left(\frac{\left(\frac{\alpha  \mu ^{\mu } \Gamma (\alpha  \mu )}{\Gamma (\mu ) \vartheta _{\nu } \hat{r}^{\alpha  \mu }}\right)^L}{2 \, \Gamma \left(\frac{L \alpha  \mu }{2 \vartheta _{\nu }}+1\right)}\right)^{- \frac{2 \vartheta _{\nu }}{\alpha  \mu  L}}
\end{align}
is coding gain for the ASER.

\subsubsection{Outage Probability}
\label{sec: Outage Probability}
The OP is the probability that the instantaneous SNR falls below a certain
threshold, $\gamma_{\text{out}}$, i.e.,
\begin{align}
    \label{eq: OP}
    P_{\text{out}_\nu}  \triangleq \text{Pr} \left[ \Psi_\nu \leq \gamma_{\text{out}} \right] = \int_{0}^{\gamma_{\text{out}}}  f_{\Psi_\nu}(\psi) \text{d}\psi.
\end{align}

Since the OP is the CDF of the instantaneous SNR evaluated at $\gamma_{\text{out}}$, we have from \eqref{eq: CDF Phi}
\begin{align}
    \label{eq: OP 2}
    \nonumber P_{\text{out}_\nu} \triangleq & F_{\Psi_\nu}(\gamma_{\text{out}})\\
    =&  \left(\frac{\alpha  \mu ^{\mu }}{\vartheta _{\nu } \Gamma (\mu ) \hat{r}^{\alpha  \mu }}\right)^L \sum _{i=0}^{\infty } \frac{\delta_i \left(\frac{ E_s}{N_0 L \gamma_{\text{out}} }\right)^{-\frac{1}{2} (\alpha  i+\alpha  \mu  L)}}{ \Gamma (\frac{\alpha  i+\alpha  \mu  L}{\vartheta _{\nu }} +1)}.
\end{align}

Similarly, as for the ASEP, we address the high-SNR regime for the OP by taking the first term in the summation of \eqref{eq: OP 2}. Thus, considering only this term, the asymptotic OP can be attained as 
\begin{align}
    \label{eq: Asymp OP}
    P_{\text{out}_\nu}  \simeq & \left( \mathcal{C}_{\text{out}_\nu} \frac{E_s}{N_0}  \right)^{-\mathcal{D}_{\text{out}_\nu}},
\end{align}
where $\mathcal{D}_{\text{out}_\nu} =(\alpha  \mu  L)/2$ is the diversity gain and
\begin{align}
    \label{eq: cod gain OP} \mathcal{C}_{\text{out}_\nu}  = \frac{1}{L \gamma _{\text{out}}} \left(\Gamma \left(\frac{L \alpha  \mu }{\vartheta _{\nu }}+1\right) \left(\frac{\alpha  \mu ^{\mu } \Gamma (\alpha  \mu )}{\Gamma (\mu ) \vartheta _{\nu } \hat{r}^{\alpha  \mu }}\right)^L\right)^{-\frac{2}{\alpha  \mu  L}}
\end{align}
is coding gain for the OP.

% %=================================================
\begin{figure}[t!]
\begin{center}
\includegraphics[trim={0cm 0cm 0cm 0cm},clip,scale=0.435]{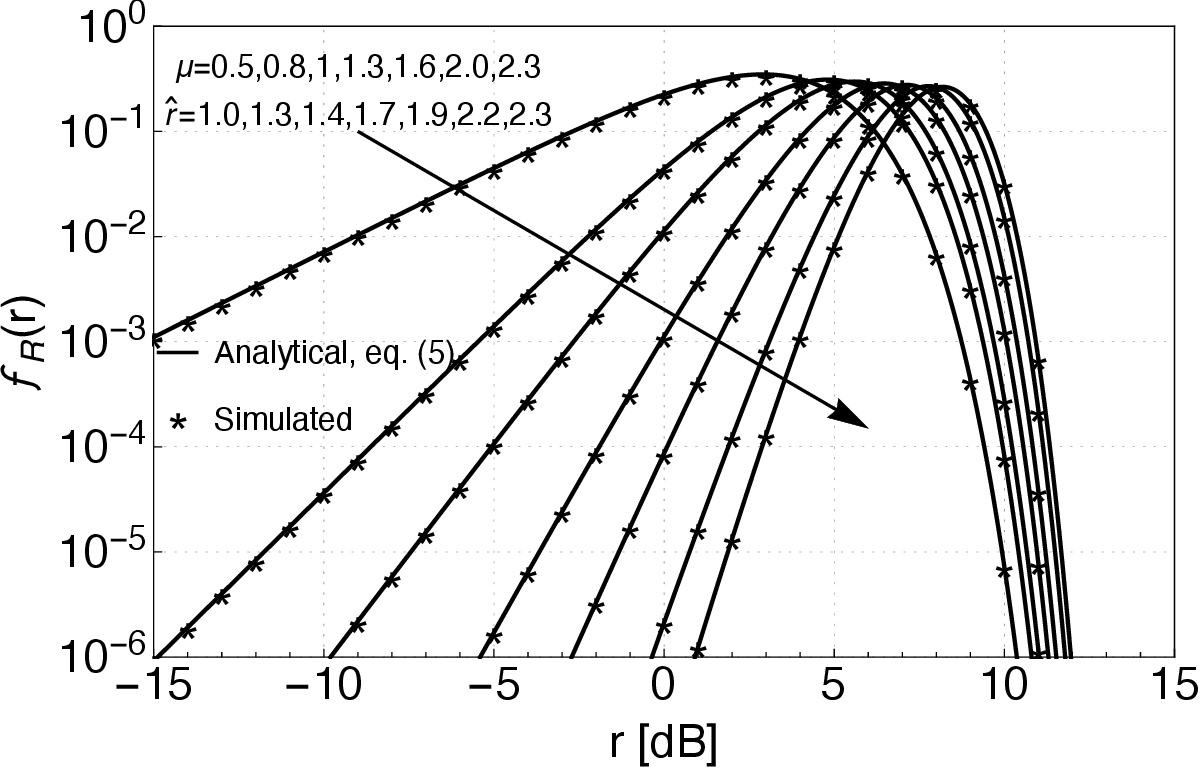}
\vspace{0cm}
\caption{PDF of $R$ for $\alpha=1.7$, $L=3$, and a range of values of $\mu$ and $\hat{r}$.}
\label{fig: PDF Z}
\end{center} 
\end{figure}
%=========================================================
%=========================================================
\begin{figure}[t!]
\begin{center}
\includegraphics[trim={0cm 0cm 0cm 0cm},clip,scale=0.435]{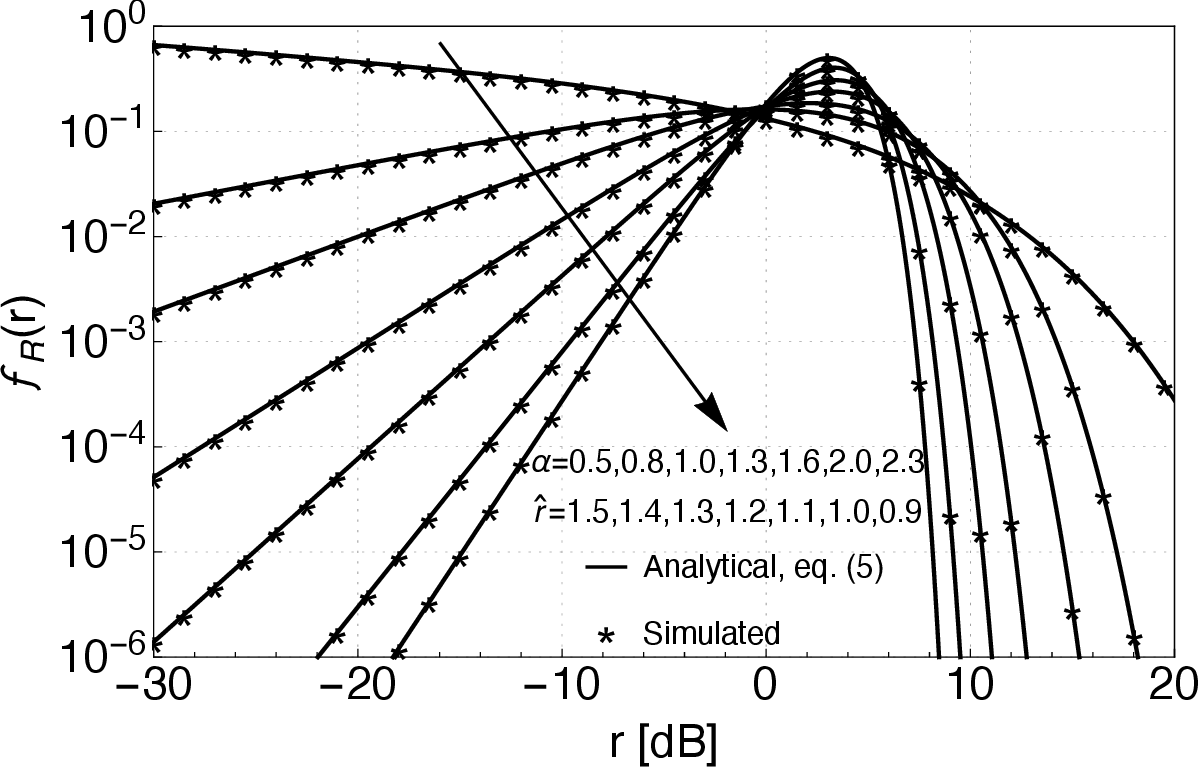}
\vspace{0cm}
\caption{PDF of $R$ for $\mu=1.7$, $L=3$, and a range of values of $\alpha$ and $\hat{r}$.}
\label{fig: PDF2 Z}
\end{center} 
\end{figure}
%=========================================================
%=========================================================
\begin{figure}[t!]
\begin{center}
\includegraphics[trim={0cm 0cm 0cm 0cm},clip,scale=0.435]{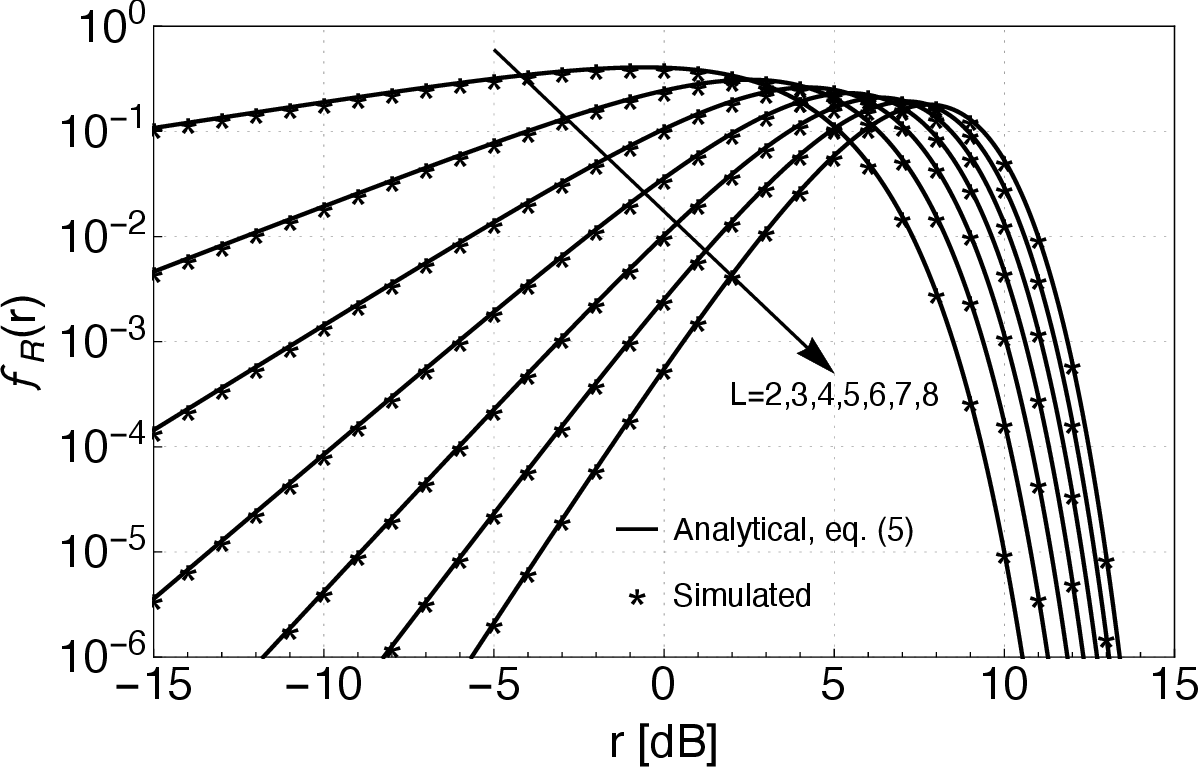}
\vspace{0cm}
\caption{PDF of $R$ for $\alpha=0.5$, $\mu=1.5$, $\hat{r}=1$, and a range of values of $L$.}
\label{fig: PDF3 Z}
\end{center} 
\end{figure}
%=========================================================
%=========================================================
\begin{figure}[t]
\begin{center}
\includegraphics[trim={0cm 0cm 0cm 0cm},clip,scale=0.435]{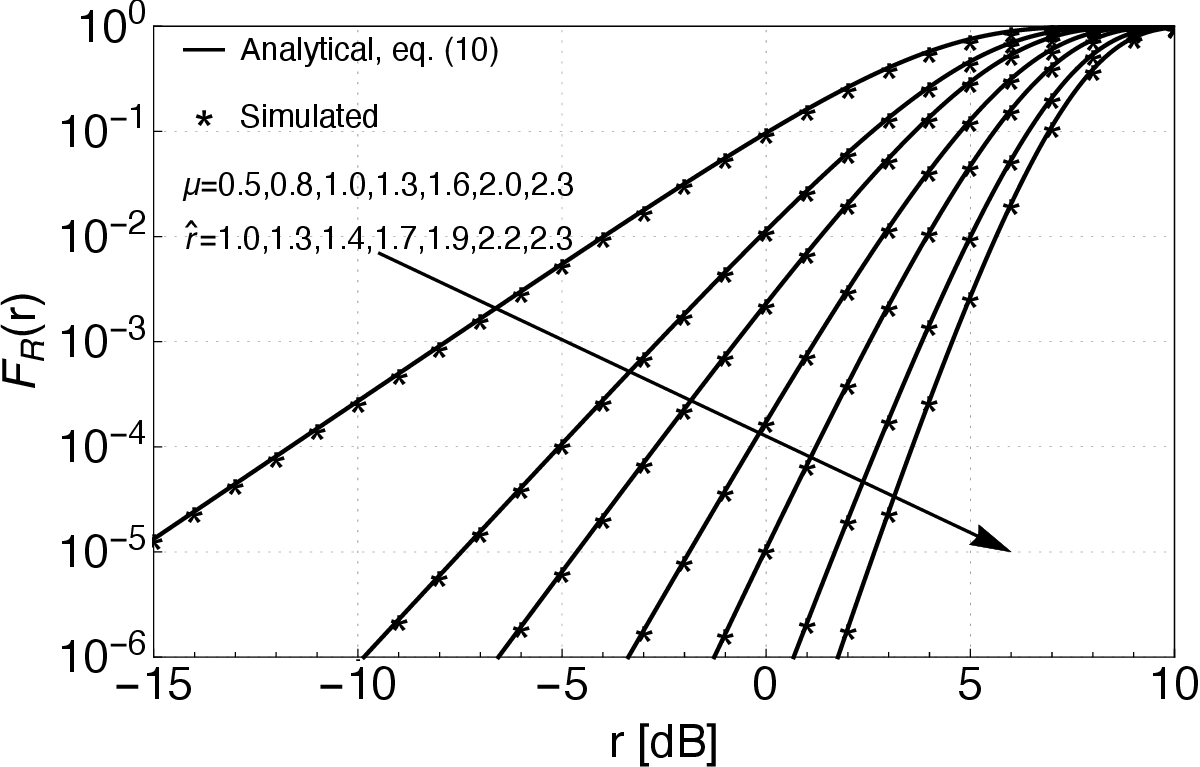}
\vspace{0cm}
\caption{CDF of $R$ for $\alpha=1.7$, $L=3$, and a range of values of $\mu$ and $\hat{r}$.}
\label{fig: CDF Z}
\end{center} 
\end{figure}
%=========================================================
%=========================================================
\begin{figure}[t!]
\begin{center}
\includegraphics[trim={0cm 0cm 0cm 0cm},clip,scale=0.435]{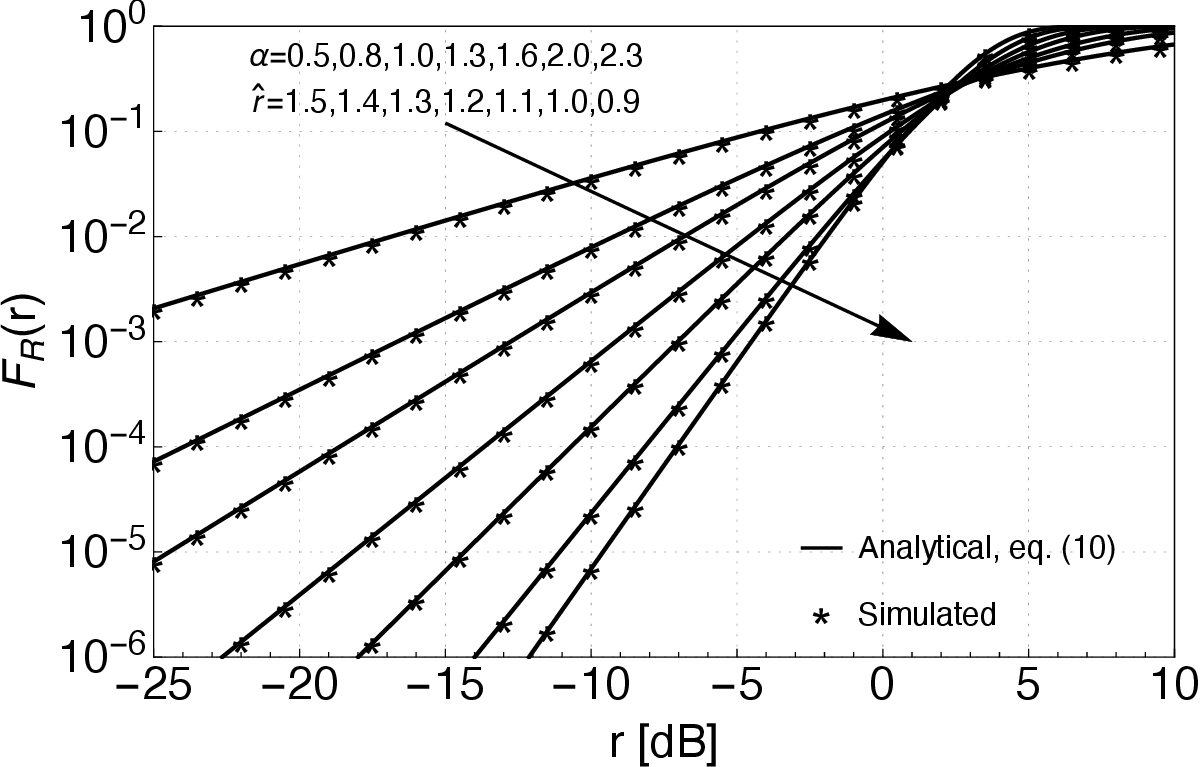}
\vspace{0cm}
\caption{CDF of $R$ for $\mu=1.7$, $L=3$, and a range of values of $\alpha$ and $\hat{r}$.}
\label{fig: CDF2 Z}
\end{center} 
\end{figure}
%=========================================================
%=========================================================
\begin{figure}[t!]
\begin{center}
\includegraphics[trim={0cm 0cm 0cm 0cm},clip,scale=0.435]{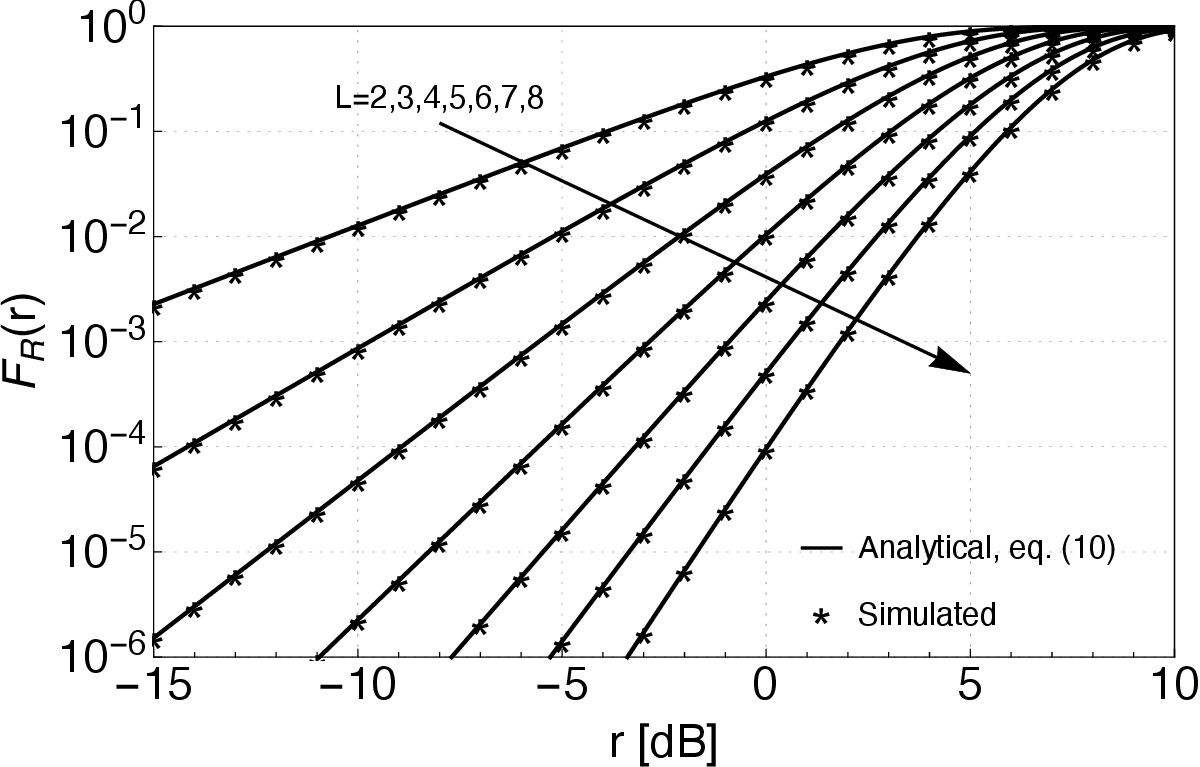}
\vspace{0cm}
\caption{CDF of $R$ for $\alpha=0.5$, $\mu=1.5$, $\hat{r}=1$, and a range of values of $L$.}
\label{fig: CDF3 Z}
\end{center} 
\end{figure}
%=========================================================

It is important to highlight that the exact and asymptotic solutions presented here for the ASER and OP are also original contributions. 
In fact, our expressions \eqref{eq: Pe final}, \eqref{eq: Pe Asymptotic}, \eqref{eq: OP 2}, and \eqref{eq: Asymp OP} are much simpler and enjoy a lower computational burden than those in \cite[eq. (28)]{Costa08aplha}, \cite[eq. (30)]{Payami21}, \cite[eq. (31)]{Payami21}, \cite[Table VII]{ElAyadi14}, and \cite[eq. (13)]{BenIssaid18}.

\section{Numerical Results and Discussion}
\label{sec: Numerical Results}

Now, we corroborate our derived expressions using numerical integration.\footnote{The PDF and the CDF of the sum $R$ were obtained by numerically evaluating the multi-fold Brennan's integral~\cite{Brennan59}.}
Additionally, we illustrate the efficiency of \eqref{eq: Final PDF} as compared to the state-of-the-art exact solutions for the sum PDF of Nakagami-$m$ \cite[eq. (4)]{Rahman11} and $\alpha$-$\mu$ \cite[eq. (3)]{Kong2021} RVs.
Since those works do not provide any exact expressions for the corresponding sum CDF, we only compare \eqref{eq: Final CDF} with the solution obtained via Brennan's integral.

Considering arbitrary values of $\alpha$, $\mu$, $\hat{r}$, and $L$, Figs. \ref{fig: PDF Z}--\ref{fig: PDF3 Z} show the analytical and simulated PDFs of $R$, while Figs. \ref{fig: CDF Z}--\ref{fig: CDF3 Z} show the  analytical and simulated CDFs of $R$. %
The values of the parameters $\alpha$, $\mu$, and $\hat{r}$ were chosen to illustrate a variety of shapes that the PDFs and CDFs may bear.
Note how well our derived expressions match the simulation results, thereby validating our findings.
%=================================================
\begin{table*}[h!]
\centering
\caption{Efficiency of~\eqref{eq: Final PDF} as compared to \cite[eq. (4)]{Rahman11} (Relative error $<10^{-6}$)}
\label{tab: Table2}
\begin{tabular}{c c c c c}
\toprule
Parameter settings & $f_R (r)$ & Elapsed time for~\eqref{eq: Final PDF} [s]  & Elapsed time for \cite[eq. (4)]{Rahman11} [s] & Time saving [\%]\\ \hline
$\alpha=2$, $\mu=1$, $\hat{r}=1$, $L=3$, $r=5$ & $0.0054$ & 0.0587 &  3.2286 & 98.1807 \\ \hline
$\alpha=2$, $\mu=1$, $\hat{r}=1$, $L=4$, $r=5$ & $0.1222$  & 0.0533 & 3.6523 & 98.5392 \\ \hline
$\alpha=2$, $\mu=2$, $\hat{r}=1$, $L=5$, $r=5$ & $0.3079$ & 0.0456 & 22.111 & 99.7938 \\ \hline
$\alpha=2$, $\mu=2$, $\hat{r}=1$, $L=6$, $r=5$ & $0.3475$ & 0.0316 & 17.112 & 99.8148 \\ \hline
$\alpha=2$, $\mu=3$, $\hat{r}=1$, $L=7$, $r=5$ & $0.5572$ & 0.0338 & 52.963 & 99.9361 \\ 
\bottomrule  
\end{tabular}
\end{table*}
%=================================================
%=================================================
\begin{table*}[t!]
\centering
\caption{Efficiency of~\eqref{eq: Final PDF} as compared to \cite[eq. (3)]{Kong2021} (Relative error $<10^{-6}$)}
\label{tab: Table3}
\begin{tabular}{c c c c c}
\toprule
Parameter settings & $f_R (r)$ & Elapsed time for~\eqref{eq: Final PDF} [s] & Elapsed time  for \cite[eq. (3)]{Kong2021} [s] & Time saving [\%] \\ \hline
$\alpha=0.5$, $\mu=2.5$, $\hat{r}=5$, $L=3$, $r=5$ & 0.0306 & 0.0549 & 2.7704 & 98.0182 \\ \hline
$\alpha=1.0$, $\mu=2.5$, $\hat{r}=5$, $L=4$, $r=15$ & 0.0572 & 0.0389 & 3.5229 & 98.8952 \\ \hline
$\alpha=1.5$, $\mu=2.5$, $\hat{r}=5$, $L=5$, $r=15$ & 0.0111 & 0.0579 & 18.866 & 99.6931 \\ \hline
$\alpha=2.0$, $\mu=2.5$, $\hat{r}=5$, $L=6$, $r=20$ & 0.0069 & 0.0990 & 82.174 & 99.8795 \\ \hline
$\alpha=2.5$, $\mu=2.5$, $\hat{r}=5$, $L=7$, $r=22$ & 0.0002 & 0.9186 & 6657.1 & 99.9862 \\ 
\bottomrule  
\end{tabular}
\end{table*}
%=================================================
%=================================================
\begin{table*}[t!]
\centering
\caption{Efficiency of~\eqref{eq: Final CDF} as compared to Brennan's formulation (Relative error $<10^{-6}$)}
\label{tab: CDF}
\begin{tabular}{c c c c c}
\toprule
Parameter settings & $F_R (r)$ & Elapsed time for~\eqref{eq: Final CDF} [s] & \begin{tabular}[c]{@{}c@{}} Elapsed time for numerical\\  integration [s] \end{tabular} & Time saving [\%] \\ \hline
$\alpha=0.5$, $\mu=2.5$, $\hat{r}=5$, $L=3$, $r=5$ & 0.0729 & 0.0132 & 1.0253 & 98.7122 \\ \hline
$\alpha=1.0$, $\mu=2.5$, $\hat{r}=5$, $L=4$, $r=15$ & 0.2235 & 0.0361 & 3.4463 &  98.9521 \\ \hline
$\alpha=1.5$, $\mu=2.5$, $\hat{r}=5$, $L=5$, $r=15$ & 0.0158 & 0.0420 & 10.4397 &  99.5972 \\ \hline
$\alpha=2.0$, $\mu=2.5$, $\hat{r}=5$, $L=6$, $r=20$ & 0.0311 &  0.0828 & 63.3356 & 99.8692  \\ \hline
$\alpha=2.5$, $\mu=2.5$, $\hat{r}=5$, $L=7$, $r=22$ & 0.0101 & 0.1086 & 537.762 & 99.9798  \\ 
\bottomrule  
\end{tabular}
\end{table*}
%=================================================
%=========================================================
\begin{figure}[t!]
\begin{center}
\includegraphics[trim={0cm 0cm 0cm 0cm},clip,scale=0.435]{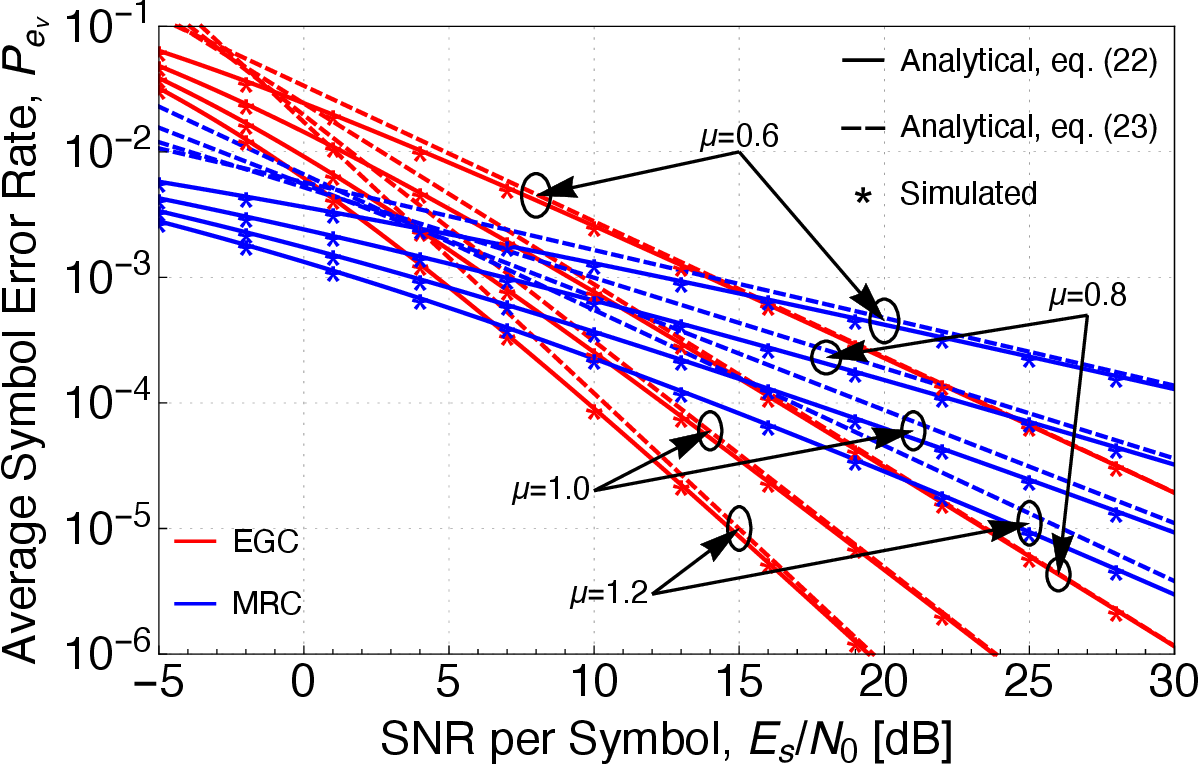}
\vspace{0cm}
\caption{ASER versus SNR per symbol for $\alpha=1.2$, $\hat{r}=2$, $L=3$, $\mathcal{G}=1$, and a range of values of $\mu$.}
\label{fig: ASER1}
\end{center} 
\end{figure}
%=========================================================
%=========================================================
\begin{figure}[t!]
\begin{center}
\includegraphics[trim={0cm 0cm 0cm 0cm},clip,scale=0.435]{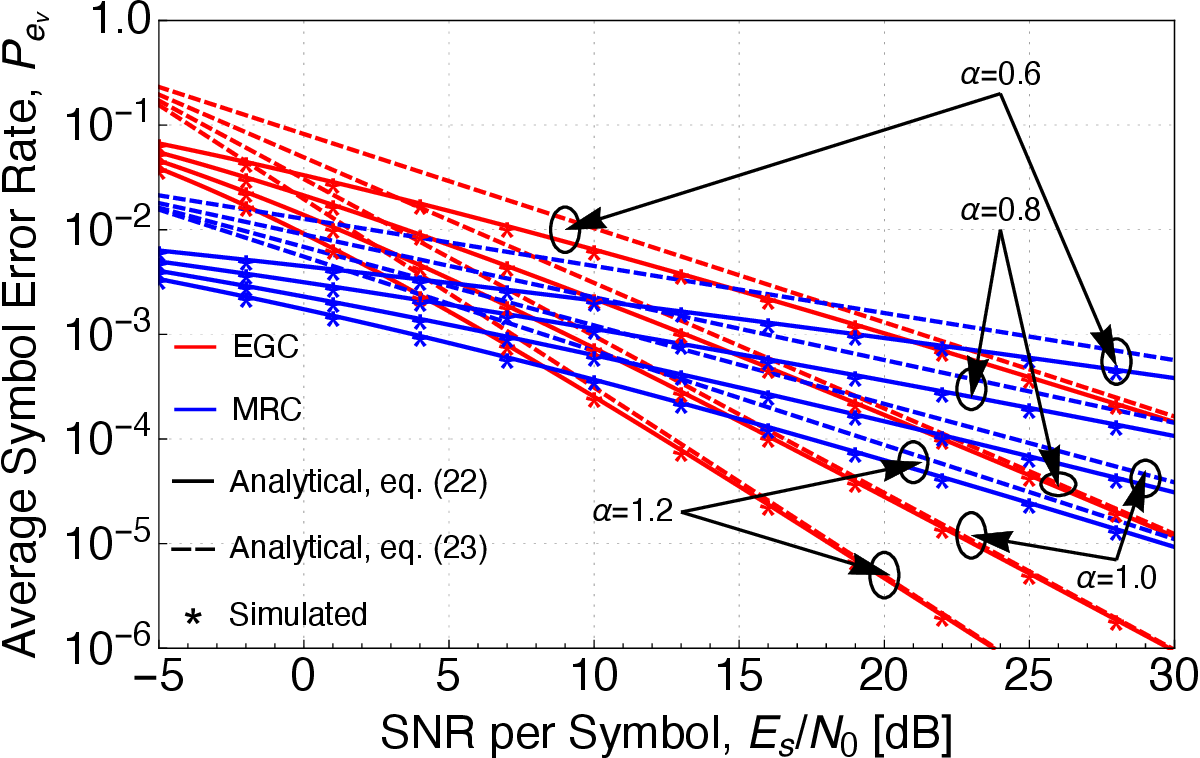}
\vspace{0cm}
\caption{ASER versus SNR per symbol for $\mu=0.9$, $\hat{r}=2$, $L=3$, $\mathcal{G}=1$, and a range of values of $\alpha$.}
\label{fig: ASER2}
\end{center} 
\end{figure}
%=========================================================
%=========================================================
\begin{figure}[t]
\begin{center}
\includegraphics[trim={0cm 0cm 0cm 0cm},clip,scale=0.435]{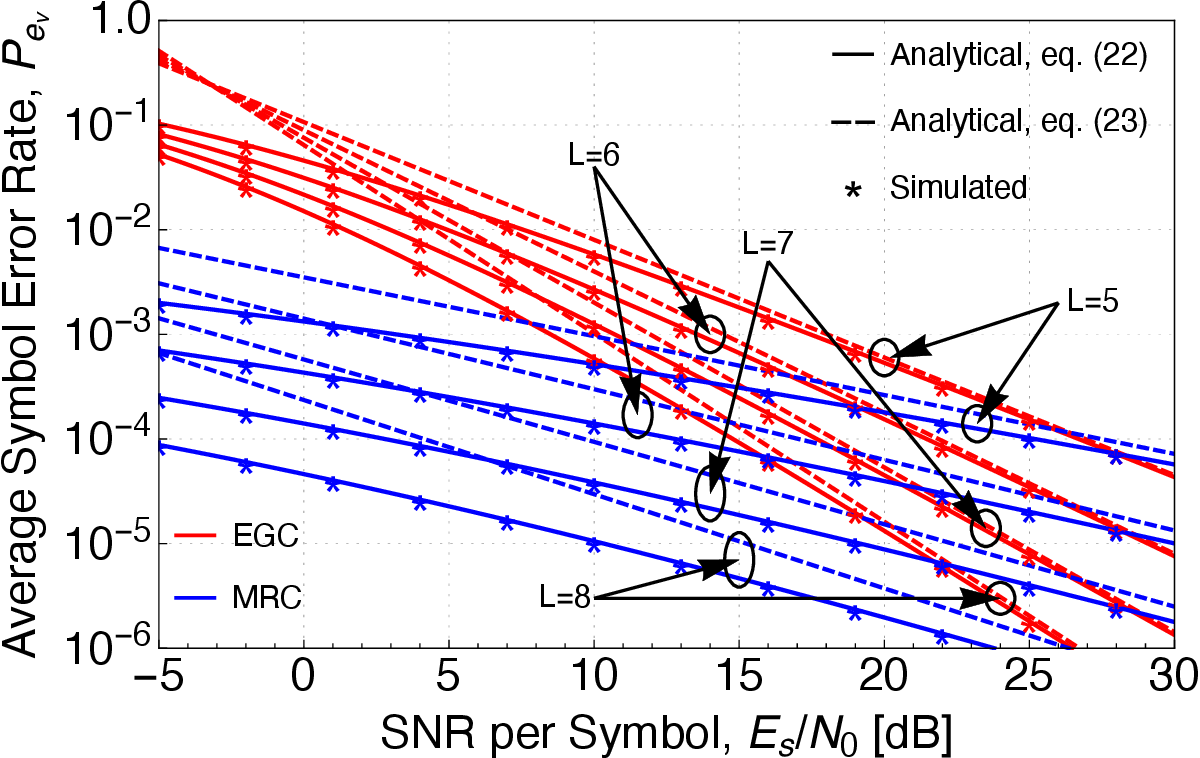}
\vspace{0cm}
\caption{ASER versus SNR per symbol for $\mu=0.5$, $\alpha=0.9$, $\hat{r}=1$, $\mathcal{G}=1$, and a range of values of $L$.}
\label{fig: ASER3}
\end{center} 
\end{figure}
%=========================================================
%=========================================================
\begin{figure}[t]
\begin{center}
\includegraphics[trim={0cm 0cm 0cm 0cm},clip,scale=0.435]{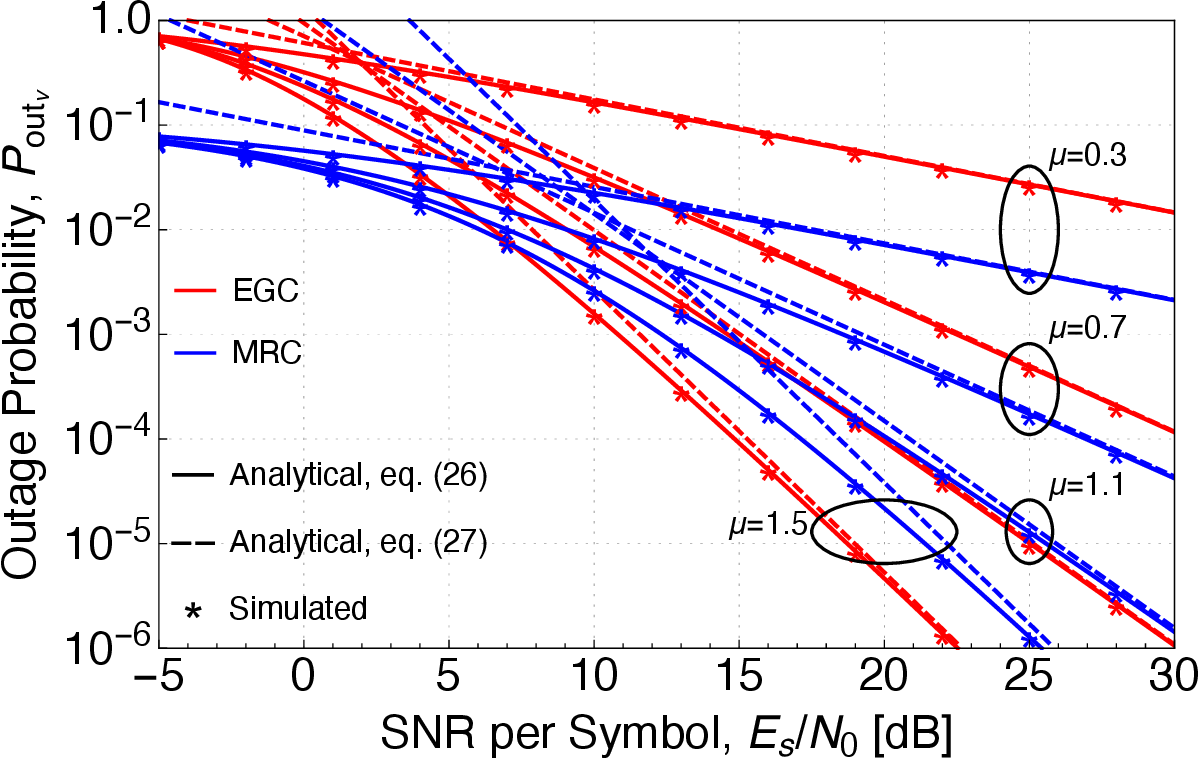}
\vspace{0cm}
\caption{OP versus SNR per symbol for $\alpha=1.2$, $\hat{r}=3$, $L=3$, ${\gamma_{\text{out}}=10}$~dB, and a range of values of $\mu$.}
\label{fig: OP1}
\end{center} 
\end{figure}
%=========================================================
%=========================================================
\begin{figure}[t]
\begin{center}
\includegraphics[trim={0cm 0cm 0cm 0cm},clip,scale=0.435]{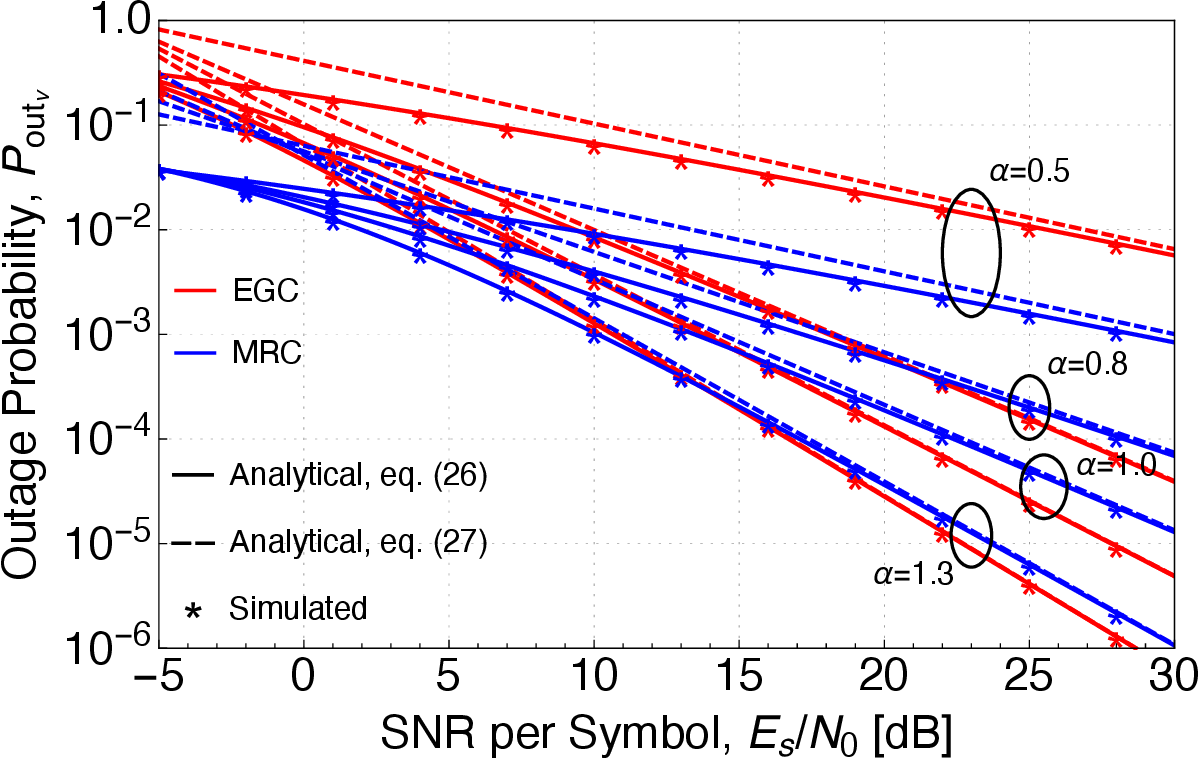}
\vspace{0cm}
\caption{OP versus SNR per symbol for $\mu=0.8$, $\hat{r}=6$, $L=3$, ${\gamma_{\text{out}}=10}$~dB, and a range of values of~$\alpha$.}
\label{fig: OP2}
\end{center} 
\end{figure}
%=========================================================
%=========================================================
\begin{figure}[t]
\begin{center}
\includegraphics[trim={0cm 0cm 0cm 0cm},clip,scale=0.435]{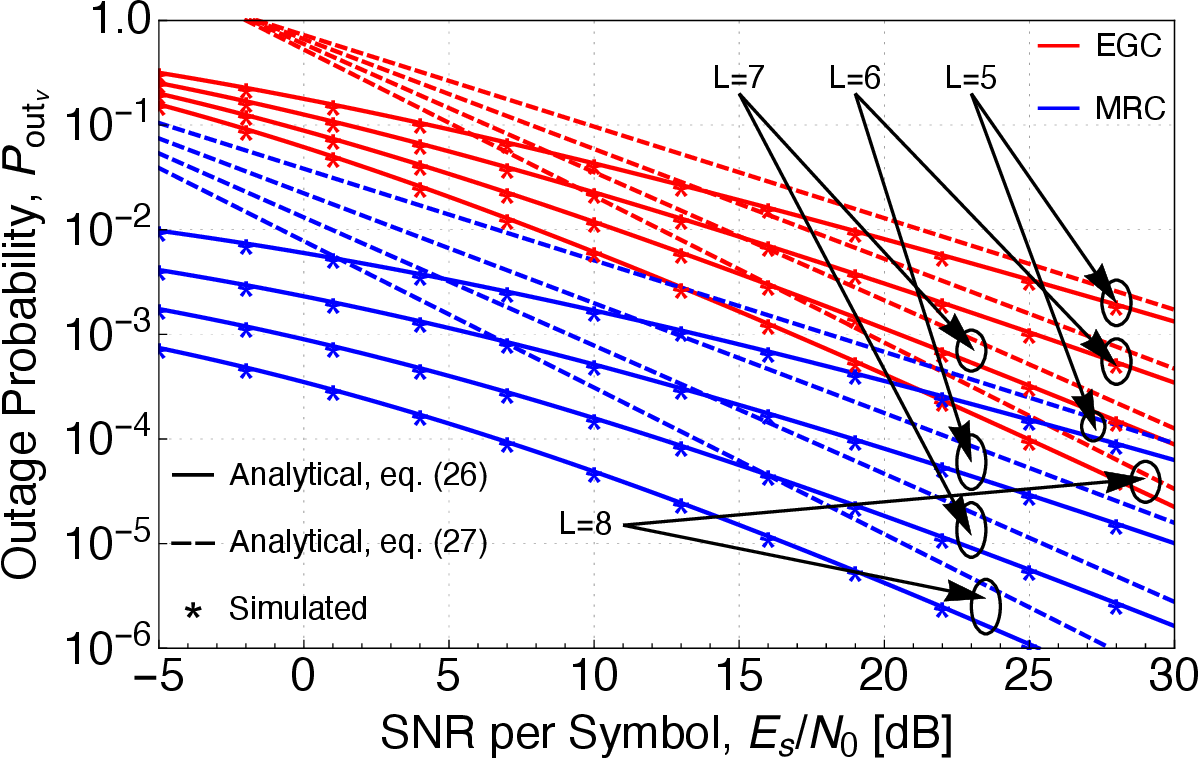}
\vspace{0cm}
\caption{OP versus SNR per symbol for $\mu=0.7$, $\alpha=0.5$, $\hat{r}=3$, ${\gamma_{\text{out}}=1}$~dB, and a range of values of $L$.}
\label{fig: OP3}
\end{center} 
\end{figure}
%=========================================================

Tables \ref{tab: Table2}--\ref{tab: CDF} show an efficiency analysis by comparing \eqref{eq: Final PDF} and \eqref{eq: Final CDF} with the exact formulations in \cite{Rahman11}, \cite{Kong2021}, and \cite{Brennan59}.
% The efficiency of \eqref{eq: Final PDF} is determined by comparing its computation time with those
% in \cite{Rahman11} and \cite{Kong2021}.
More specifically, Tables \ref{tab: Table2} and \ref{tab: Table3} show the elapsed time for \eqref{eq: Final PDF} as compared to \cite[eq. (4)]{Rahman11} and \cite[eq. (3)]{Kong2021}, respectively. 
Each computation time was recorded when the corresponding formulation reached a desired relative error less than $10^{-6}$.
Note in all the tables that our expressions provide remarkable time savings of above $98$\%. 
More importantly, notice that the elapsed time of the proposed expressions exhibit a small variation over $L$. In other words, the efficiency of our formulations is bearly dependent on the number of summands. On the other hand, the elapsed times for  \cite[eq. (4)]{Rahman11} and \cite[eq. (3)]{Kong2021} tend to grow exponentially with~$L$. Similar improvements are provided by the sum CDF in~\eqref{eq: Final CDF} when compared to Brennan's formulation, as can be observed from Table \ref{tab: CDF}.

Figs. \ref{fig: ASER1}--\ref{fig: ASER3} show the ASER as a function of the SNR per symbol for different values of $\mu$, $\alpha$, and $L$. 
The results indicate that the ASER for EGC and MRC receivers decreases as either $\mu$, $\alpha$, or $L$ increases. This is because the system's diversity order, $\alpha \mu L/2 \vartheta _{\nu}$, increases with any of these parameters. 
Interestingly, notice in Figs. \ref{fig: ASER1} and \ref{fig: ASER2} show that depending on the $\alpha$ and $\mu$ parameters, the performance of EGC surpasses the performance of MRC in the high SNR regime, showing a lower ABER for a given SNR per symbol.
Once again, observe that our derived ASER expressions match perfectly the numerical simulations.

Figs. \ref{fig: OP1}--\ref{fig: OP3} show the OP in terms of the SNR per symbol for different values of $\mu$, $\alpha$, and $L$. 
In all the figures, note how the OP decreases as $\mu$, $\alpha$, or $L$ increases.
In addition, as in the ASER case, observe in Figs. \ref{fig: OP1} and \ref{fig: OP2} that depending on the $\alpha$ and $\mu$ parameters, the performance of EGC overcomes the performance of MRC in the high SNR regime.
Finally, notice how our exact curves coincide perfectly with the numerical simulations.

\section{Conclusion}
\label{sec: Conclusions}
In this paper, we provided new, handy, exact formulations for the sum PDF and the sum CDF of i.i.d. $\alpha$-$\mu$ RVs. 
Our results are arguably the most efficient and tractable formulations to date. It is noteworthy that, as opposed to the available solutions, whose tractability becomes impracticable with the increase of the number of RVs, our solution' complexity is independent of the number of summands.
As an application example, we analyzed the performance of $L$-branch pre-detection EGC and MRC receivers operating over $\alpha$-$\mu$ fading channels. With this goal, we provided exact and asymptotic expressions for the ASER and the OP, which so far have been investigated in terms of time-consuming or approximate~solutions.  In future works, we aim to explore our proposed analysis considering a broad range of emergent scenarios, such as those in the presence of high co-channel interference. Additionally, an extension of this framework applied to i.n.i.d. shall be tackled.

\begin{appendices}
\section{Proof of \eqref{eq: Final PDF}}
\label{app: PDF R}
To find the PDF of the sum in \eqref{eq: sum}, we begin by taking the Laplace transform of the marginal PDF in \eqref{eq: PDF Xn}, i.e.,
\begin{align}
    \label{eq: Laplace Def}
    \nonumber \Laplace \left\{f_{R_n} \right\} (s) & \triangleq \mathbb{E} \left[\exp  \left( -s R_n\right)\right] \\
    & = \int _0^{\infty }\exp \left( -s r_n\right) f_{R_n}(r_n) \text{d}r_n,
\end{align}
\normalsize
where $s\in\mathbb{C}$. 
Replacing \eqref{eq: PDF Xn} into \eqref{eq: Laplace Def}, along with some algebraic manipulations, we obtain
\begin{align}
    \label{eq: Meijer G}
    \nonumber \Laplace \left\{f_{R_n}\right\} (s) =& \frac{\alpha  \mu ^{\mu }}{\Gamma (\mu ) \hat{r}^{\alpha  \mu }} \int _0^{\infty }\exp (-s r_n) r_{n}^{\alpha  \mu -1} \\
     & \times  G_{1,0}^{0,1} \left[ \begin {array} {c} - \\0 \\\end {array} \left| - \mu  \left(\frac{r_n}{\hat{r}}\right)^{\alpha }\right. \right] \text{d}r_n,
\end{align}
\normalsize
where $G_{u,v}^{p,q} \left[ \cdot \right]$ is the Meijer $G$-function~\cite[eq. (16.17.1)]{Olver10}. 

With the aid of \cite[eq. (07.34.02.0001.01)]{Mathematica} and after interchanging the order of integration, we rewrite \eqref{eq: Meijer G} as
\begin{align}\label{eq:factorized_contour}
    \nonumber \Laplace \left\{f_{R_n}\right\} (s) =& \frac{\alpha  \mu ^{\mu }}{\Gamma (\mu ) \hat{r}^{\alpha  \mu }} \left( \frac{1}{2 \pi  j} \right)  \oint_{\mathfrak{L}_{\textbf{s},1}} \Gamma \left(s_1\right) \left(\mu  \left(\frac{1}{\hat{r}}\right)^{\alpha }\right)^{-s_1} \\
    & \times \int_0^{\infty } \exp (-s r_n) r_{n}^{\alpha  \mu -1-\alpha  s_1} \text{d}r_n \ \text{d}s_1,
\end{align}
\normalsize
where $s_1$ is a complex variable of integration, and $\mathfrak{L}_{\textbf{s},1}$ is a closed contour in the complex plane that encloses all the poles of $\Gamma(s_1)$~\cite{mathai09}. 
The change in the order of integration is allowed since  $\int _0^{\infty } \left|\exp (-s r_n) f_{R_n}(r_n)\right| \text{d}r_n < \infty$.
After evaluating the resulting inner integral in \eqref{eq:factorized_contour}, we obtain
\begin{align}
    \label{eq: Complex int1}
    \Laplace \left\{f_{R_n}\right\} (s) =& \frac{\alpha  \mu ^{\mu }}{\Gamma (\mu ) \left(\hat{r} s\right)^{\alpha  \mu }} \left( \frac{1}{2 \pi  j} \right) \oint_{\mathfrak{L}_{\textbf{s},1}^\dagger} \Phi \left(s_1 \right) \text{d}s_1,
\end{align}
\normalsize
where the integration kernel $\Phi \left(s_1 \right)$ is given as
\begin{align}
    \label{}
    \Phi \left(s_1 \right)= \Gamma \left(s_1\right) \Gamma \left(\alpha  \left(\mu -s_1\right)\right) \left(\mu  \left(\frac{1}{\hat{r} s}\right)^{\alpha }\right)^{-s_1},
\end{align}
and $\mathfrak{L}_{\textbf{s},1}^\dagger$ is a new contour that appears as the integration over $r_n$ deforms $\Phi \left(s_1 \right)$.
This new contour allows us to represent $\Laplace \left\{f_{R_n}\right\} (s)$ in terms of a meromorphic function analytically defined on the strip $0<s_1<\mu$ and with singularities located at $s_1=-i$ and $s_1=(\alpha  \mu +i)/\alpha$, $\forall i \in \mathbb{N}_0$.
For convenience, we define $\mathfrak{L}_{\textbf{s},1}^\dagger$ as a contour that starts at the point $\infty + j \xi _1$ and ends at the point $\infty + j \xi _2$, where $-\infty<\xi _1<\xi _2<+\infty$ such that $\mathfrak{L}_{\textbf{s},1}^\dagger$ encloses all the poles of $\Gamma \left(\alpha  \left(\mu -s_1\right)\right)$ in the positive~direction.

Now, using Cauchy's residue theorem, \eqref{eq: Complex int1} can be expressed as follows~\cite{Kreyszig10}:
\begin{align}
    \label{eq: Res Def}
    \Laplace\left\{f_{R_n}\right\} (s) = \frac{\alpha  \mu ^{\mu }}{\Gamma (\mu ) \left(\hat{r} s\right)^{\alpha  \mu }} \sum _{i=0}^{\infty }  \mathcal{R} \left[ \Phi \left(s_1 \right);\left\{s_1=-i\right\}\right],
\end{align}
\normalsize
where $\mathcal{R}\left[ \Phi \left(s_1 \right);\left\{s_1=-i\right\}\right]$ is the residue of $\Phi \left(s_1 \right)$ at the poles $s_1=-i$.
Then, using the residue operation~\cite[eq. (16.3.3)]{Kreyszig10}, we have
\begin{align}
    \label{eq: Sum Residues Xn}
    \nonumber \Laplace \left\{f_{R_n} \right\} (s) = & \frac{\alpha  \mu ^{\mu }}{\Gamma (\mu ) \left(\hat{r} s\right)^{\alpha  \mu }} \\
    & \times \sum _{i=0}^{\infty } \frac{\Gamma (\alpha  (i+\mu )) \left(-\mu  \left(\frac{1}{\hat{r} s}\right)^{\alpha }\right)^i}{i!}.
\end{align}
\normalsize

Since the RVs $\{R_n\}_{n=1}^L$ are mutually independent, the sum PDF is given~by
\begin{equation}\label{eq:conv}
   f_{R}(r) =  f_{R_1}(r_1) * f_{R_2}(r_2) * \cdots  * f_{R_L}(r_L).
\end{equation}
\normalsize
In addition, as $\{R_n\}_{n=1}^L$ are i.i.d. RVs, the Laplace transform of \eqref{eq:conv} yields~\cite{papoulis02}
\begin{align}
    \label{eq: Laplace Power}
    \Laplace  \left\{ f_{R}\right\}(s) & = \left[ \Laplace  \left\{f_{R_n}\right\} (s) \vphantom{2^{2}}\right]^L.
\end{align}
\normalsize
Accordingly, the inverse Laplace transform of \eqref{eq: Laplace Power} is given by~\cite{bookSchiff99}
\begin{align}
    \label{eq: Inverse Laplace Def}
    \nonumber f_{R}(r) \triangleq & \Laplace^{-1}  \left\{ \Laplace \left\{ f_{R}\right\}(s)\right\}(r) \\
    =& \left(\frac{1}{2 \pi j} \right) \oint_{\mathfrak{L}_{\textbf{s},2}}  \exp \left( s r \right)  \Laplace  \left\{ f_{R}\right\}(s)  \text{d}s,
\end{align}
\normalsize
where $\mathfrak{L}_{\textbf{s},2}$ is the Bromwich contour. Substituting \eqref{eq: Sum Residues Xn} into \eqref{eq: Laplace Power}, and then this result into~\eqref{eq: Inverse Laplace Def}, we obtain
\begin{align}
    \label{eq: Inverse Laplace s}
    \nonumber f_{R}(r) = & \left(\frac{\alpha  \mu ^{\mu }}{\Gamma (\mu ) \hat{r}^{\alpha  \mu }}\right)^L \left(\frac{1}{2 \pi j} \right) \oint_{\mathfrak{L}_{\textbf{s},2}} \exp \left( s z \right) s^{-L \alpha  \mu }\\
    & \times \left(\sum _{i=0}^{\infty } s^{-\alpha i} a_i \right)^L \text{d} s,
\end{align}
\normalsize
where 
\begin{align}
    \label{}
    a_i = \frac{\Gamma (\alpha  (i+\mu )) \left(-\mu  \left(\frac{1}{\hat{r}}\right)^{\alpha }\right)^i}{i!}.
\end{align}

From \eqref{eq: Inverse Laplace s}, we let
\begin{align}
    \label{eq: Differentiation}
    \left( \sum _{i=0}^{\infty}   s^{- \alpha i} a_i \right)^L =\sum _{i=0}^{\infty}s^{-\alpha i} \delta_i.
\end{align}
\normalsize
Now, we employ the following differential equation:
\begin{align}
    \label{eq: differential equation}
    \zeta \left(\zeta^{L }\right) '=L \zeta^{L} \zeta^{'},
\end{align}
where the apostrophe denotes the derivative with respect to $s^{- \alpha}$, and 
\begin{align}
    \label{}
    \zeta= &\sum _{i=0}^{\infty}   s^{- \alpha i} a_i \\
    \zeta^{L} =& \sum _{i=0}^{\infty}s^{-\alpha i} \delta_i.
\end{align}
After solving \eqref{eq: differential equation}, it immediately follows that the coefficients $\delta_i$ can be obtained by \eqref{eq: Coefficients}. 

Replacing \eqref{eq: Differentiation} and \eqref{eq: Coefficients} into \eqref{eq: Inverse Laplace s}, we get
\begin{align}
    \label{eq: Inverse Laplace s step 2}
    \nonumber f_{R}(r) = & \left(\frac{\alpha  \mu ^{\mu }}{\Gamma (\mu ) \hat{r}^{\alpha  \mu }}\right)^L \left(\frac{1}{2 \pi j} \right) \oint_{\mathfrak{L}_{\textbf{s},2}} \exp \left( s z \right) s^{-L \alpha  \mu }\\
    & \times \sum _{i=0}^{\infty}s^{-\alpha i} \delta_i \ \text{d} s.
\end{align}

Notice that \eqref{eq: Inverse Laplace s step 2} has now a pole of order $(L \alpha \mu+ \alpha i)$ at $s=0$.  
Therefore, if $(L \alpha \mu+ \alpha i) \in \mathbb{Z}^{+}$, then  $\mathfrak{L}_{\textbf{s},2}$ is the standard Bromwich contour defined in~\cite[eq. (4.4)]{bookSchiff99}.
On the other hand, if $(L \alpha \mu+ \alpha i)\notin \mathbb{Z}^{+}$, then \eqref{eq: Inverse Laplace s step 2} is a multi-valued function. In this case, the Bromwich contour must be modified in order to make \eqref{eq: Inverse Laplace s step 2} a single-valued function. To do so, we define a branch point at $s=0$ and a branch cut along the negative real axis.
Accordingly, we now define $\mathfrak{L}_{\textbf{s},2}$ as a loop that begins and ends at $-\infty$, encircling the pole $s=0$ once in the positive~direction. 
Notice that for both cases, $(L \alpha \mu+ \alpha i) \in \mathbb{Z}^{+} $ or $(L \alpha \mu+ \alpha i) \notin \mathbb{Z}^{+} $, the selected contours for $\mathfrak{L}_{\textbf{s},2}$ enclose the pole $s=0$. 
Thus, by applying the residue theorem, \eqref{eq: Inverse Laplace s step 2} can be rewritten as
\begin{align}
    \label{eq: residue step 2}
        f_{R}(r) = & \left(\frac{\alpha  \mu ^{\mu }}{\Gamma (\mu ) \hat{r}^{\alpha  \mu }}\right)^L \mathcal{R} \left[ \Theta (s);(s=0)\right],
\end{align}
where 
\begin{align}
    \label{eq: integration kernel 2}
    \Theta (s)=   \sum _{i=0}^{\infty} \delta_i  s^{-\alpha i-L \alpha  \mu } \exp \left( s z \right) 
\end{align}
is the integration kernel of \eqref{eq: Inverse Laplace s step 2}.

Finally, since the residue operation is a \textit{linear mapping}, then we can perform a term-by-term inversion with the aid of $\mathcal{R} \left[ s^{-a} \exp \left(s z \right);(s=0)\right]=z^{a-1}/ \Gamma (a)$ \cite{bookSchiff99}, resulting in \eqref{eq: Final PDF}. This completes the proof.

\section{Absolute Convergence of \eqref{eq: Final PDF}}
\label{app: Absolute Convergence}

Herein, we show the absolute convergence of~\eqref{eq: Final PDF}. To this end, we have to prove that~\cite{Kreyszig10}
\begin{align}    \label{eq: absolute value}
     \sum _{i=0}^{\infty } \left| \frac{\delta_i r^{\alpha  i+\alpha  \mu  L-1}}{\Gamma (\alpha  i+\alpha  \mu  L)} \right|< \infty,
\end{align}
or, equivalently,
\begin{align}    \label{eq: absolute delta i}
     \sum _{i=0}^{\infty }  \frac{ \left|\delta_i \right| r^{\alpha  i+\alpha  \mu  L-1}}{\Gamma (\alpha  i+\alpha  \mu  L)}< \infty.
\end{align}

An upper bound for $|\delta_i|$, $\forall i\geq 1$, can be obtained by taking the absolute value of the summands, i.e.,
\begin{align}
    \label{eq: abs delta}
    \left|\delta_i\right|  \leq & \frac{1}{i \Gamma (\alpha  \mu)} \sum _{l=1}^i \frac{\left|\delta_{i-l}\right| \left|l L+l-i \right| \Gamma (\alpha  (l+\mu )) \left(\mu  \left(\frac{1}{\hat{r}}\right)^{\alpha }\right)^l}{l!}.
\end{align}
Notice that depending on $\alpha$, \eqref{eq: abs delta} can be either a decreasing function (if $0<\alpha < 1$) or an increasing function (if $\alpha \geq 1$). Hence, we must obtain a bound for each condition.

For the case of $0<\alpha < 1$, $\left|\delta_{i}\right|$ can be bounded as follows:
\begin{align}
    \label{eq: abs delta k<1}
    \nonumber \left|\delta_i\right| \overset{(a)}{\leq} &  \frac{\mu \left|\delta_{i-1}\right| \left|L+1-i \right| \Gamma (\alpha  (1+\mu ))  }{i \Gamma (\alpha  \mu)}   \left(\frac{1}{\hat{r}}\right)^{\alpha } \\
    \nonumber \overset{(b)}{<} & 2 \mu L \left|\delta_{i-1}\right|  \Gamma (\alpha  (1+\mu ))    \left(\frac{1}{\hat{r}}\right)^{\alpha } \\
    \overset{(c)}{=} & r^{\alpha  \mu  L-1}   \Gamma (\alpha  \mu )^L  E_{\alpha ,\alpha  \mu  L}\left(2 \mu  L \left(\frac{r}{\hat{r}}\right)^{\alpha } \Gamma (\mu  \alpha +\alpha )\right),
\end{align}
where in step (a) we used the fact that  $\left|\delta_{i}\right|<\left|\delta_{i-1}\right|$, in step (b) we employed that  $\left| -i+1+ L\right| \leq i L$ and that ${1/\Gamma (\alpha  \mu )<2}$,~$\forall (\alpha,\mu)$, and in step (c) we recursively solved the coefficients $\left| \delta_{i-1}\right|$.

Now, with the aid of \eqref{eq: abs delta k<1}, an upper bound for \eqref{eq: absolute delta i} can be obtained as
\begin{align}
    \label{eq: reulting series k<1}
    \nonumber \sum _{i=0}^{\infty } & \frac{\left|\delta_i\right| r^{\alpha  i+\alpha  \mu  L-1}}{\Gamma (\alpha  i+\alpha  \mu  L)} \overset{(a)}{<} r^{\alpha  \mu  L-1}   \Gamma (\alpha  \mu )^L  \\
    \nonumber & \times \left[ \frac{1}{\Gamma (L \alpha  \mu )} + \sum _{i=1}^{\infty } \frac{ \left(2 \mu  L \Gamma (\mu  \alpha +\alpha ) \left(\frac{r}{\hat{r}}\right)^{\alpha }\right)^i}{\Gamma (i \alpha +L \mu  \alpha )} \right] \\
    & \overset{(b)}{=} r^{\alpha  \mu  L-1}   \Gamma (\alpha  \mu )^L  E_{\alpha ,\alpha  \mu  L}\left(2 \mu  L \left(\frac{r}{\hat{r}}\right)^{\alpha } \Gamma (\mu  \alpha +\alpha )\right),
\end{align}
where in step (a) we separated the first term of the sum and used \eqref{eq: abs delta k<1}, and in step (b) we employed the Mittag-Leffler function denoted by $E_{(\cdot,\cdot)}(\cdot)$~\cite[eq. (10.46.3)]{Olver10}.

For the case of $\alpha \geq 1$, $\left|\delta_{i}\right|$ can be bounded as
\begin{align}
    \label{eq: abs delta k>1}
    \left|\delta_i\right|  \nonumber \overset{(a)}{<} & 2 L \sum _{l=1}^i \frac{\left|\delta_{i-l}\right| \Gamma (\alpha  (l+\mu )) \left(\mu  \left(\frac{1}{\hat{r}}\right)^{\alpha }\right)^l}{l!} \\
    \overset{(b)}{\leq} &   \frac{2  L \, \Gamma (\alpha  \mu )^L \Gamma (\alpha  (l+\mu )) \left(\mu  \left(\frac{1}{\hat{r}}\right)^{\alpha }\right)^i}{i!}
\end{align}
where in step (a) we employed that  $\left| -i+l+l L\right| \leq i L$, for $1 \leq l \leq i$, and that ${1/\Gamma (\alpha  \mu )<2}$,~$\forall (\alpha,\mu)$, and in step (b) we only used the last term of the sum (i.e., $l=i$).

From \eqref{eq: abs delta k>1}, an upper bound for \eqref{eq: absolute delta i} can be attained as
\begin{align}
    \label{eq: reulting series k>1}
    \nonumber \sum _{i=0}^{\infty } & \frac{\left|\delta_i\right| r^{\alpha  i+\alpha  \mu  L-1}}{\Gamma (\alpha  i+\alpha  \mu  L)} \overset{(a)}{<} r^{\alpha  \mu  L-1}  \Gamma (\alpha  \mu )^L  \\
    \nonumber  \times & \left[\frac{1}{\Gamma (L \alpha  \mu )} +2 L \sum _{i=1}^{\infty } \frac{\Gamma (\alpha  (i+\mu )) \left(\mu  \left(\frac{r}{\hat{r}}\right)^{\alpha }\right)^i}{i! \Gamma (i \alpha +L \mu  \alpha )}\right]\\
    \nonumber \overset{(b)}{\leq} & r^{\alpha  \mu  L-1}  \Gamma (\alpha  \mu )^L  \left[\frac{1}{\Gamma (L \alpha  \mu )} +2 L \sum _{i=1}^{\infty } \frac{\left(\mu  \left(\frac{r}{\hat{r}}\right)^{\alpha }\right)^i}{i!}\right]\\
    \overset{(c)}{=} &  r^{\alpha  \mu  L-1}  \Gamma (\alpha  \mu )^L \left[ \frac{1}{\Gamma (L \alpha  \mu )} +2 L \exp \left(\mu  \left(\frac{r}{\hat{r}}\right)^{\alpha }\right)-2 L \right],
\end{align}
where in step (a) we separated the first term of the sum and used \eqref{eq: abs delta k>1}, in step (b) the relationship $\Gamma (\alpha  (i+\mu ))/ \Gamma (\alpha  i+\alpha  \mu  L) \leq 1$ was applied, and in step (c) we used \cite[eq.~4.2.1]{abramowitz72}.

Finally, as \eqref{eq: reulting series k<1} and \eqref{eq: reulting series k>1} exist and are finite, then 
\eqref{eq: Final PDF} converges absolutely $\forall (\alpha, \mu, \hat{r})$, which completes the proof.

\section{The $\alpha$-$\mu$ Distribution as a Check}
\label{app: singel alpha-mu}
Using $L=1$ and separating the first term of the sum in \eqref{eq: Final PDF}, we obtain 
\begin{align}
    \label{eq: Chi 1}
    f_R (r)=\frac{\alpha  \mu ^{\mu } r^{\alpha  \mu -1}}{\Gamma (\mu ) \hat{r}^{\alpha  \mu }} \left( \frac{\delta _0}{\Gamma (\alpha  \mu )}+ \sum _{i=1}^{\infty } \frac{\delta _i r^{\alpha  i}}{\Gamma (i \alpha +\mu  \alpha )}\right),
\end{align}
where the coefficients $\delta_i$ now reduce to
\begin{subequations}
\label{eq: Reduced Coefficients Chi}
\begin{align}
    \delta_0=& \Gamma (\alpha  \mu )\\
    \delta_i=& \frac{\Gamma (\alpha  (i+\mu )) \left(-\mu  \hat{r}^{-\alpha }\right)^i}{\Gamma (i+1)}, \ \ \ \ \ i \geq 1.
\end{align}
\end{subequations}

Finally, substituting \eqref{eq: Reduced Coefficients Chi} into \eqref{eq: Chi 1}, and applying \cite[eq. (5.2.11.11)]{prudnikov98Vol1}, we obtain
\begin{align}
    \label{eq: Pre-Final Chi}
    f_R (r)= \frac{\alpha  \mu ^{\mu } r^{\alpha  \mu -1}}{\Gamma (\mu ) \hat{r}^{\alpha  \mu }} \exp \left(-\mu  \left(\frac{r}{\hat{r}}\right)^{\alpha }\right),
\end{align}
which is the $\alpha$-$\mu$ PDF.

\section{Truncation Bounds for the Sum PDF}
\label{app: Truncation Bounds}

The truncation error in \eqref{eq: truncation def PDF} can be bounded by taking the absolute values of the summands as follows:
\begin{align}
    \label{app: bound PDF}
    \epsilon_f < \left(\frac{\alpha  \mu ^{\mu }}{\Gamma (\mu ) \hat{r}^{\alpha  \mu }}\right)^L \sum _{i=\mathcal{N}_{\text{T}}}^{\infty } \frac{ \left| \delta _i \right|  r^{\alpha  i+\alpha  \mu  L-1}}{\Gamma (i \alpha +L \mu  \alpha )}.
\end{align}

If $0<\alpha<1$, \eqref{app: bound PDF} can be further bounded as
\begin{align}
    \label{}
    \nonumber \epsilon_f \overset{(a)}{<} & \, r^{\alpha  \mu  L-1} \left(\frac{\alpha  \mu ^{\mu } \Gamma (\alpha  \mu )}{\Gamma (\mu ) \hat{r}^{\alpha  \mu }}\right)^L \\
    \nonumber  & \times \sum _{i=\mathcal{N}_{\text{T}}}^{\infty } \frac{\left(2 \mu  L \left(\frac{r}{\hat{r}}\right)^{\alpha } \Gamma (\mu  \alpha +\alpha )\right)^i}{\Gamma (i \alpha +L \mu  \alpha )} \\
    \overset{(b)}{<} & \, \mathcal{B}_{f}^\dagger,
\end{align}
where in step (a) we used the upper bound in \eqref{eq: abs delta k<1}, and in step (b) we employed the series representation of the Mittag-Leffler function \cite{Olver10} (refer to \eqref{eq: bound k<1}).

On the other hand, if $\alpha \geq 1$, \eqref{app: bound PDF} can be further bounded as
\begin{align}
    \label{}
    \nonumber \epsilon_f \overset{(a)}{<}  & \frac{2 L}{r} \left(\frac{\left(\alpha  \mu ^{\mu } \Gamma (\alpha  \mu )\right) \left(\frac{r}{\hat{r}}\right)^{\alpha  \mu }}{\Gamma (\mu )}\right)^L \\
   \nonumber  & \times \sum _{i=\mathcal{N}_{\text{T}}}^{\infty } \frac{\Gamma (\alpha  (i+\mu )) \left(\mu  \left(\frac{r}{\hat{r}}\right)^{\alpha }\right)^i}{i! \Gamma (i \alpha +L \mu  \alpha )} \\
   \nonumber  \overset{(b)}{<}  & \frac{2 L}{r} \left(\frac{\left(\alpha  \mu ^{\mu } \Gamma (\alpha  \mu )\right) \left(\frac{r}{\hat{r}}\right)^{\alpha  \mu }}{\Gamma (\mu )}\right)^L \sum _{i=\mathcal{N}_{\text{T}}}^{\infty } \frac{\left(\mu  \left(\frac{r}{\hat{r}}\right)^{\alpha }\right)^i}{i!}\\
    \overset{(c)}{<} & \, \mathcal{B}_{f}^*,
\end{align}
where in step (a) we used the upper bound in \eqref{eq: abs delta k>1}, in step (b) we used the fact that $\Gamma (\alpha  (i+\mu ))/\Gamma (i \alpha +L \mu  \alpha ) \leq 1$, and in step (c) we employed the series representation of the lower incomplete gamma function~\cite[eq. (06.06.06.0002.01)]{Mathematica} along with some algebraic manipulations (refer to \eqref{eq: bound k>=1}).

\section{Proof of \eqref{eq: Final CDF}}
\label{app: CDF R}
The CDF of the sum in \eqref{eq: sum} can be found from \eqref{eq: Final PDF} as
\begin{align}
    \label{eq: CDF 1}
    \nonumber F_{R}(r) \triangleq& \int_{0}^{r}  f_{R}(\upsilon) \text{d} \upsilon \\ 
    = &\left(\frac{\alpha  \mu ^{\mu }}{\Gamma (\mu ) \hat{r}^{\alpha  \mu }}\right)^L \int_{0}^{r}  \sum _{i=0}^{\infty } \frac{\delta_i \upsilon^{\alpha  i+\alpha  \mu  L-1}}{\Gamma (i \alpha +L \mu  \alpha )} \text{d} \upsilon.
\end{align}
Due to the absolute convergence of \eqref{eq: Final PDF} (vide Appendix \ref{app: Absolute Convergence}), we can interchange the order of integration in \eqref{eq: CDF 1}, resulting in \eqref{eq: Final CDF}. This completes the proof.

\end{appendices}
\bibliographystyle{IEEEtran}
\bibliography{Bibliography}

\end{document}